\documentclass[11pt]{article}
\usepackage{amsmath,amssymb,amsthm,bm}
\usepackage{verbatim}
\usepackage{graphicx}
\usepackage{color}
\usepackage[round]{natbib}
\usepackage{pstricks, pst-node, pst-plot}

\newcommand{\eps}{\varepsilon}
\newcommand{\R}{\mathbb{R}}

\newcommand{\mH}{\mathcal{H}}

\renewcommand{\P}{{\rm P}}
\newcommand{\E}{{\rm E}}
\newcommand{\var}{{\rm var}}

\newcommand{\FWER}{{\rm FWER}}
\newcommand{\FDR}{{\rm FDR}}
\newcommand{\Power}{{\rm Power}}
\newcommand{\Bon}{{\rm Bon}}
\newcommand{\BH}{{\rm BH}}
\newcommand{\argmax}{\operatornamewithlimits{argmax}}
\newcommand{\mad}{{\rm MAD}}
\newcommand{\med}{{\rm Med}}

\theoremstyle{plain}
\newtheorem{thm}{Theorem}[section]
\newtheorem{lemma}[thm]{Lemma}
\newtheorem{prop}[thm]{Proposition}

\theoremstyle{definition}
\newtheorem{defn}[thm]{Definition}
\newtheorem{proc}[thm]{Procedure}
\newtheorem{example}[thm]{Example}

\begin{document}

\title{Peak Detection as Multiple Testing}
\author{
Armin Schwartzman$^1$, Yulia Gavrilov$^1$, Robert J. Adler$^2$ \\[.5cm]
$^1$ Department of Biostatistics,
Harvard School of Public Health \\
$^2$ Department of Electrical Eng.,
Technion - Israel Institute of Technology
}
\date{\today}
\maketitle

\begin{abstract}
This paper considers the problem of detecting equal-shaped non-overlapping unimodal peaks in the presence of Gaussian ergodic stationary noise, where the number, location and heights of the peaks are unknown. A multiple testing approach is proposed in which, after kernel smoothing, the presence of a peak is tested at each observed local maximum. The procedure provides strong control of the family wise error rate and the false discovery rate asymptotically as both the signal-to-noise ratio (SNR) and the search space get large, where the search space may grow exponentially as a function of SNR. Simulations assuming a Gaussian peak shape and a Gaussian autocorrelation function show that desired error levels are achieved for relatively low SNR and are robust to partial peak overlap. Simulations also show that detection power is maximized when the smoothing bandwidth is close to the bandwidth of the signal peaks, akin to the well-known matched filter theorem in signal processing. The procedure is illustrated in an analysis of electrical recordings of neuronal cell activity.
\end{abstract}

\section{Introduction}

Peak detection is a common statistical problem in the analysis of high-throughput data. Examples include identification of binding sites on the genome \citep{MACS:2008}, DNA sequencing \citep{Li:2000}, identification of proteins in mass spectrometry \citep{Morris:2005,Harezlak:2008}, detection of action potentials in neuronal recordings \citep{Baccus:2002}, detection of heart beats in electrocardiograms \citep{Arzeno:2008}, and identification of signatures in galaxy spectra \citep{Brutti:2007}. A crucial step in the analysis of these data, both for dimension reduction and consequent inference, is the detection of an unknown number of signal peaks with a temporal or spatial structure in the presence of background noise. The main challenge in these problems is that both the number of peaks and their location are unknown. In addition, the data often consist of a single long sequence with no replicates.

While there are many peak detection algorithms in the scientific literature, they tend to be geared toward specific applications and their performance is often evaluated empirically. In particular, peak detection algorithms often require a threshold, but the choice of the threshold is ad-hoc and does not take into account the error inflation produced by multiple testing. Errors in peak detection can lead to erroneous conclusions in later steps of the data analysis. There is a need to approach peak detection from a formal statistical viewpoint. Our objective is to develop a general statistical procedure for identifying signal peaks in the presence of background noise with proven error control, while at the same time being easy to implement and efficient to run on large data sets.

In this paper, we consider the specific problem of detecting equal-shaped non-overlapping unimodal peaks in the presence of Gaussian stationary noise, where the number, locations and heights of the peaks are unknown. We have chosen this particular setting for its analytical tractibility, but we believe that it captures the essence of the general peak detection problem and can serve as a theoretical basis for future applications in the particular scientific disciplines.

The assumption of not knowing the number of peaks is key. If the number of peaks were known, then the unknown locations could be estimated solving a nonlinear least-squares problem \citep{OBrien:1994,Li:2000,Li:2004}. The main difficulty of this approach is that not knowing the number of peaks implies not knowing the number of location parameters to be estimated. As a consequence, the problem becomes akin to the model selection problem in regression \citep{Li:2000,Li:2004}. Alternative solutions using $L_1$ regularization include direct penalization of the estimated signal \citep{OBrien:1994} and a modification of the LASSO algorithm where the penalty is applied to the difference between consecutive coefficients to account for the ordered structure of the data \citep{Tibshirani:2005}.

Rather than an estimation problem, we view peak detection as a multiple testing problem, where, at each of a set of locations, a test is performed for whether the signal is nonzero at that location. Our approach can be viewed essentially as a search for peaks over the length of the data. The idea has been used elsewhere \citep{Yasui:2003,Morris:2005,Chumbley:2009} but not formally for multiple testing, and is motivated as follows. A full search for peaks would require testing every single observed point for significance. However, if the peaks are assumed unimodal and non-overlapping, dimensionality can be reduced dramatically by testing only at locations that resemble peaks, that is, local maxima of the observed sequence. In addition, it is known that signal-to-noise ratio (SNR) can be improved by local smoothing such as averaging over a local neighbourhood. Based on these ideas, our proposed algorithm consists of the following steps:
\begin{enumerate}
\item {\em Kernel smoothing}
\item {\em Candidate peaks}: find local maxima of the smoothed sequence.
\item {\em P-values}: compute a p-value at each local maximum, defined as the probability of peering the observed intensity of the local maximum or higher according to the distribution that would be expected if we only observed noise.
\item {\em Multiple testing}: use a multiple testing procedure to find a global threshold and declare significant all peaks exceeding that threshold.
\end{enumerate}
For Step 4, we consider two standard multiple testing procedures: Bonferroni and Benjamini-Hochberg (BH) \citep{Benjamini:1995}. Our peak detection algorithm has the advantage of being simple, easy to remember and efficient to implement for large data sets. At the same time, we believe it to be powerful and we show that it provides guaranteed global error control. As measures of global error we consider both family-wise error rate (FWER), to be controlled by the Bonferroni procedure, and false discovery rate (FDR) \citep{Benjamini:1995}, to be controlled by the BH procedure. For simplicity, we concentrate on positive signals and one-sided tests, but this is not crucial.

Our proofs of error control assume that the noise is a continuous ergodic stationary Gaussian process. This assumption permits a closed form formula for computing the p-values corresponding to local maxima of the observed process. The distribution of the height of a local maximum of a Gaussian process is not Gaussian but has a heavier tail, and its computation requires careful conditioning based on the calculus of Palm probabilities \citep{Cramer:1967,Adler:2010}. This is crucial to ensure that p-values are valid.

Another interesting and challenging aspect of the proof of error control is the fact that the number of tests, which is equal to the number of local maxima in a given interval, is a random quantity. Proofs of error control in the multiple testing literature usually assume that the number of tests is fixed. In our proofs, we overcome this difficulty using an asymptotic argument for large search space, so that the error behaves approximately as it would if the number of tests were equal to its expected value.

In the proofs, the asymptotics for large search space are combined with asymptotics for large SNR. The large SNR assumption helps solve the issue of identifiability of peaks, as it implies that each signal peak is represented by only one observed local maximum with probability tending to one. The asymptotic rates, however, are such that the search space is allowed to grow much faster than the SNR; exponentially faster. In this sense, we do not consider the large SNR assumption restrictive.

For concreteness, we conduct simulations assuming a Gaussian peak shape and a Gaussian autocorrelation function. Our simulations confirm that moderate values of SNR are enough to provide desired error levels. Our simulations also show that the performance is maintained under a substantial amount of overlap between neighboring peaks.

Defining detection power as the expected fraction of true peaks detected, we prove that the peak detection algorithm is consistent in the sense that its power tends to one under the above asymptotic conditions. We then use simulations to address the question of optimal bandwidth. In the above Gaussian autocorrelation model, we find that the detection power is maximized when the smoothing bandwidth is close to the bandwidth of the signal peaks, slightly larger when the noise is uncorrelated and getting smaller as the amount of autocorrelation in the noise increases. This result is similar to the well-known matched filter theorem in signal processing, which states that the SNR after smoothing is maximized when the smoothing kernel matches the shape and width of the signal \citep{Simon:1995,Pratt:1991}. Most noticeably, the optimal bandwidth for peak detection as multiple testing is different, in fact much larger, than the usual optimal bandwidth for nonparametric estimation. A distinction between our setting and the matched filter theorem is that in many applications of the latter in digital communications and radar screening, the time of testing is prespecified, while in our case it is random.

We illustrate our procedure with a data set of neural electrical recordings, where the objective is to detect action potentials representing cell activity \citep{Baccus:2002}. The data analysis takes advantage of the sparsity of the signal in order to estimate the parameters of the noise process.

The rest of the paper is organized as follows. Section \ref{s:theory} presents the theoretical results. Section \ref{s:sim} presents the simulation results. Section \ref{s:data} presents the data example. Section \ref{s:discussion} summarizes. Proofs are given in Section \ref{s:proofs}.
All the simulations and the data analysis were implemented in \texttt{R}.

\section{Theory}
\label{s:theory}
\subsection{The model}
\label{sec:model}
Consider the signal-plus-noise model
\begin{equation}
\label{eq:signal+noise}
y(t) = \mu(t) + z(t), \qquad t \in \R
\end{equation}
where $\mu(t)$ is the signal we wish to detect and $z(t) \in
\mathcal{C}^2$ is stationary ergodic zero-mean Gaussian noise. The
signal $\mu(t)$ is a (sparse) train of positive peaks of the form
\begin{equation}
\label{eq:mu}
\mu(t) = \sum_{j=-\infty}^\infty a_j h_b(t - \tau_j), \qquad
h_b(t) = \frac{1}{b} h\left(\frac{t}{b}\right)
\end{equation}
where $a_j, b > 0$. The peak shape $h(t) \ge 0$ is assumed
unimodal with mode at $t=0$ and no other critical points within its support 
$S = \{t: h(t) > 0\}$,  where $S$ is assumed to be compact and connected.
Assume that $h(t)$ has unit action $\int_{-\infty}^\infty h(t)\,dt =
1$. Let 
$$
w_\gamma(t) = \frac{1}{\gamma}w \left(\frac{t}{\gamma}\right), \qquad \gamma>0
$$
 be a unimodal kernel  with compact and connected support and
$\int_{-\infty}^\infty w_\gamma(t)\,dt = 1$. Define 
\begin{equation}
\label{eq:hgamma}
h_\gamma(t) = w_\gamma(t)*h_b(t)
\end{equation}
with support $S_\gamma =  \{t: h_\gamma(t) > 0\}$.
We assume that $h_\gamma(t)$ is unimodal with mode at some interior
point of $S$ not necessarily equal to 0, with no other critical
points. In addition we assume that $h_\gamma(t)$ is twice
differentiable in the interior of $S_\gamma$. 
Note that if, for example,  both $h(t)$ and $w(t)$ are truncated
Gaussian functions, then all the assumptions made on
$h_\gamma(t)$ above are valid.  

Finally, the smoothed signal and smoothed noise are defined as 
\begin{equation}
\label{eq:mu-gamma}
\mu_\gamma(t) = w_\gamma(t) * \mu(t) = \sum_{j=-\infty}^\infty a_j h_\gamma(t - \tau_j), \qquad z_\gamma(t) = w_\gamma(t) * z(t).
\end{equation}
We require that the supports $S_{j, \gamma} =
\{t:h_\gamma(t-\tau_j)>0\}$ do not overlap for all $j$.
In this sense the signal can be considered sparse.

\subsection{The procedure}
\label{sec:proc}
Suppose we observe $y(t)$ in the segment $T = [-L/2,L/2]$,
which contains $J_L$ peaks. The objective is to identify a set of 
locations $\tilde{\tau}_1,\ldots,\tilde{\tau}_{\tilde{m}}$ that are
close in location and in number to the true set of peak locations
$\tau_1,\ldots,\tau_{J_L}$, while controlling the probability of
obtaining false peaks. Consider the following procedure.

\begin{proc}
\label{alg:proc}
\hfill\par\noindent
\begin{enumerate}
\item {\em Kernel smoothing}:
Construct the process
\begin{equation}
\label{eq:conv}
x_\gamma(t) = w_\gamma(t) * y(t) =
\int_{-\infty}^\infty w_\gamma(t-s) y(s)\,ds, 
\end{equation}
where we ignore boundary effects at $\pm L/2$.
\item {\em Candidate peaks}:
Find all the local maxima of $x_\gamma(t)$ in $[-L/2,L/2]$, i.e. find
the set
\begin{equation}
\label{eq:T}
\tilde{T} = \left\{ t \in \left[-\frac{L}{2},\frac{L}{2}\right]: \quad
\dot{x}_{\gamma}(t) = \frac{dx_{\gamma}(t)}{dt} = 0, \quad
\ddot{x}_{\gamma}(t) = \frac{d^2 x_{\gamma}(t)}{dt^2} < 0 \right\}.
\end{equation}
\item {\em P-values}:
For each $t \in \tilde{T}$ calculate the p-value for testing the hypothesis
$$
\mH_{0}(t): \ \mu(t) = 0
\quad \text{vs.} \quad
\mH_{A}(t): \  \mu(t) > 0
$$
\item {\em Multiple testing}:
Apply a multiple testing procedure on the set of p-values and declare significant all peaks whose p-values are smaller than the corresponding threshold, which may depend on $\gamma$ and $L$.
\end{enumerate}
\end{proc}

The idea behind Procedure \ref{alg:proc} is very simple: smooth, find
local maxima, feed the list of local maxima into any standard multiple 
testing procedure. These steps are easy to remember and implement.
The first two steps are straightforward. The last two need careful
study, which we do next.

Step 3 is detailed in Section \ref{sec:pvalue} below. For Step 4, we use the Bonferroni procedure to control FWER and  the BH procedure  to control FDR. Let $\tilde{m}$ be the number of local maxima in $\tilde{T}$, i.e. the number of tested hypotheses. To apply the Bonferroni procedure at level $\alpha$ we compare each p-value to the  threshold $\alpha/\tilde{m}$ and reject all the hypotheses whose p-values are below the threshold.
To apply the BH procedure at level $\alpha$ we first order the p-values and then compare the ordered $i$-th p-value with $i\alpha/\tilde{m}$. Defining $k$ as the maximal index for which the ordered p-value is smaller than the corresponding threshold, we reject $k$ hypotheses with the $k$ smallest p-values. More details are given in the Sections below.

\subsection{P-values}
\label{sec:pvalue}
A crucial step in implementing any marginal multiple testing procedure is calculating p-values. P-values are always computed under the complete null hypothesis of no signal anywhere. For Step 3 of Procedure \ref{alg:proc}, we proceed as follows.

\begin{defn}
\label{defn:p-values}
Assume the model of Section \ref{sec:model} with $\mu(t) = 0, \forall t$, so that $x_\gamma(t) = z_\gamma(t)$, given by \eqref{eq:mu-gamma}. Let $ F_\gamma(u)$ denote the \emph{right} cumulative distribution function (cdf) of $z_\gamma(t)$ at the local maxima $t \in \tilde{T}$,
\begin{equation}
\label{eq:palm}
 F_\gamma(u) = \P\Big\{z_\gamma(t) > u ~\Big|~ t \in \tilde{T}\Big\}.
\end{equation}
Then the p-value of the observed $x_\gamma(t)$ at $t\in \tilde{T}$ is
\begin{equation}
\label{eq:p-value}
p_\gamma(t) =  F_\gamma[x_\gamma(t)], ~~~t\in\tilde{T}.
\end{equation}
\end{defn}

The probability in \eqref{eq:palm} follows the Palm distribution for maxima and we use its properties to compute an explicit formula for p-values. One needs to be careful while computing this probability to avoid bias because of two reasons. First, although we use the usual symbol for conditioning, this is not the usual conditioning event. Computing this probability by usual conditioning gives the Gaussian distribution. However, computing the probabilities just at local maxima based on the Gaussian distribution will lead to the biased down p-values. Second, we are conditioning on an event of probability zero. See \citet[Ch. 6]{Adler:2010} for a more detailed discussion. The formula for evaluating the required distribution \eqref{eq:palm} is given by Proposition \ref{thm:p-values}
below.

\begin{prop}
\label{thm:p-values}
Assume the model of Section \ref{sec:model} and the procedure of Section \ref{sec:proc}. Under the complete null hypothesis $\mu(t) = 0, \forall t$, define the moments
\begin{equation}
\label{eq:moments}
\sigma^2_\gamma = \var[z_\gamma(t)], \qquad
\lambda_{2,\gamma} = \var[\dot{z}_\gamma(t)], \qquad
\lambda_{4,\gamma} = \var[\ddot{z}_\gamma(t)].
\end{equation}
Then, 
\begin{equation}
\label{eq:distr}
 F_\gamma(u) = 1 - \Phi\left(u \sqrt{\frac{\lambda_{4, \gamma}}{\Delta}}\right) + \sqrt{\frac{2\pi\lambda^2_{2, \gamma}}{\lambda_{4, \gamma}\sigma^2_\gamma}}\phi\left(\frac{u}{\sigma_\gamma}\right)\Phi\left( u \sqrt{\frac{\lambda^2_{2, \gamma}}{\Delta \sigma_\gamma^2}}\right),
\end{equation}
where
\begin{equation}
\label{eq:delta}
\Delta =  \sigma^2_\gamma\lambda_{4,\gamma} - \lambda_{2,\gamma}^2
\end{equation}
and $\phi(x)$, $\Phi(x)$, are the standard normal density and cdf,
respectively.
\end{prop}

This result was proven by \citet[Ch. 10]{Cramer:1967}, using the well known Kac-Rice formula \citep{Rice:1945}, \citep[Ch. 11]{Adler:2007}. The proof, which we omit here, is based on other intermediate results that we will also need later, and so we state them in the next lemma. 

\begin{lemma}
\label{lemma:m_u}
Let $z(t) \in \mathcal{C}^2$ be an ergodic stationary zero-mean Gaussian process with spectral moments $\sigma^2 = \var[z(t)]$, $\lambda_2 = \var[\dot{z}(t)]$ and $\lambda_4 = \var[\ddot{z}(t)]$, defined on a compact set $ T \subset \mathbb{R} $ with non-empty interior and finite measure $|T|$. 
\begin{enumerate}
\item Define respectively the number of local maxima in $T$ and the number of local maxima in $T$ that cross the
threshold $u$ as
$$
\begin{aligned}
\tilde{m}(-\infty; T) &= \#\left\{ t \in T: \
\dot{z}(t) = 0,~ \ddot{z}(t) < 0 \right\} \\
\tilde{m}(u; T) &= \#\left\{ t \in T: \
z(t) > u,~ \dot{z}(t) = 0,~ \ddot{z}(t) < 0 \right\}.
\end{aligned}
$$
Then
$$
\begin{aligned}
\label{eq:m_u}
\E[\tilde{m}(-\infty; T)] &= |T|\E[\tilde{m}(-\infty; [0, 1])], \\
\E[\tilde{m}(u; T)] &= |T|\E[\tilde{m}(u; [0,1])],
\end{aligned}
$$
where $\E[\tilde{m}(-\infty; [0, 1])] = \sqrt{\lambda_4/\lambda_2}/(2\pi)$ is the expected number of local maxima and $\E[\tilde{m}(u; [0,1])]$ is the expected number of local maxima above the threshold $u$ in $[0, 1]$, with expression given in \citet[Ch. 10]{Cramer:1967}.  

\item
The right cdf of the heights of the local maxima of $z(t)$ is given by
\begin{equation}
\label{eq:12}
\P\Big\{z(t) > u ~\Big|~ \dot{z}(t) = 0,~ \ddot{z}(t) < 0 \Big\}  = \frac{\E[\tilde{m}(u; [0, 1])]}{\E[\tilde{m}(-\infty; [0, 1])]}.
\end{equation}
\end{enumerate}
\end{lemma}

Proposition \ref{thm:p-values} follows directly from Lemma \ref{lemma:m_u} applied to the process $z_\gamma(t)$. In particular, \eqref{eq:distr} follows directly from evaluating \eqref{eq:12}. Referring to \eqref{eq:distr}, note that, for large enough $u$, 
\begin{equation}
\label{eq:pvapprox}
F_\gamma(u) \approx  \sqrt{\frac{2\pi\lambda^2_{2, \gamma}}{\lambda_{4, \gamma}\sigma^2_\gamma}}\phi\left(\frac{u}{\sigma_\gamma}\right).
\end{equation}
Therefore, the tail of the distribution is proportional to a Gaussian density, so it behaves like a Rayleigh distribution.

A peculiar characteristic of Procedure \ref{alg:proc} is that the number of tests $\tilde{m}$, which is the same as the number of p-values, is random. In particular, under the complete null hypothesis $\mu(t) = 0, \forall t$, the expected number of tests can be computed explicitly applying Lemma \ref{lemma:m_u} to the process $z_\gamma(t)$ as
\begin{equation}
\label{eq:Em}
\E[\tilde{m}(-\infty; [-L/2, L/2])] = \frac{L}{2\pi}
\sqrt{\frac{\lambda_{4,\gamma}}{\lambda_{2,\gamma}}}.
\end{equation}

The quantities $\sigma^2_\gamma$, $\lambda_{2,\gamma}$, and
$\lambda_{4,\gamma}$ in Proposition \ref{thm:p-values} depend on the kernel $w_\gamma(t)$
and the autocorrelation of the original noise process $z(t)$.
A specific form of the p-values and the expected number of tests may
be obtained for the following Gaussian autocorrelation model.

\begin{example}[\bf{Gaussian autocorrelation model}]
\label{ex:Gaussian-ACVF}
Assume the complete null hypothesis. Suppose $w(t) = \phi(t)\mathbf{1}[-c, c ]$ is a truncated normal density and assume that the autocovariance function of the noise process $z(t)$ is proportional to a normal density, i.e.
$$
z(t) = \sigma \int_{-\infty}^\infty w_\nu(s-t)\,dB(s), \qquad
w_\nu(t) = \frac{1}{\nu} \phi\left(\frac{t}{\nu}\right)
$$
where $B(s)$ is standard Brownian motion and $\nu > 0$. Ignoring the truncation at $\pm c$, the required moments in \eqref{eq:distr} or \eqref{eq:Em} are
\begin{equation}
\label{eq:Gaussian-moments}
\sigma^2_\gamma = \frac{\sigma^2}{2\sqrt{\pi} \xi}, \qquad
\lambda_{2,\gamma} = \frac{\sigma^2}{4\sqrt{\pi} \xi^3}, \qquad
\lambda_{4,\gamma} = \frac{3\sigma^2}{8\sqrt{\pi} \xi^5}, \qquad
\xi = \sqrt{\gamma^2 + \nu^2}.
\end{equation}
Derivations are given in  Section \ref{App:2}.
Therefore, for the Gaussian autocorrelation model,
\begin{equation}
\label{eq:Gaussian-ACVF}
\begin{gathered}
 F_\gamma(u)  = 1-\Phi\left(\frac{u}{\sigma} \sqrt{3\sqrt{\pi}\xi} \right)  +  \sqrt{\frac{2\pi}{3}}\phi\left(\frac{u}{\sigma} \sqrt{2\sqrt{\pi}\xi}\right)\Phi\left( \frac{u}{\sigma} \sqrt{\sqrt{\pi}\xi}\right), \\
\E[\tilde{m}(-\infty; [-L/2, L/2])]= \frac{L}{2\pi\xi}
\sqrt{\frac{3}{2}}.
\end{gathered}
\end{equation}
If  $w(t)$ is truncated at $\pm c$ where $c$ is large enough, say
$c=4$, the moments in \eqref{eq:Gaussian-moments} are a good
approximation for the same moments computed in the truncated version.
\end{example}


\subsection{Error rates definitions}

In this subsection we define the two error rates to control, FWER and FDR. Let us define first what is considered a false discovery and what is considered a true discovery under the true model.

\begin{defn}
\label{defn:support}
Let $S_j = \{t: h_b(t-\tau_j)>0  \}$ be the support of $h_b(t-\tau_j)$. 
Define the signal region $\mathbb{S}_1$ and null region $\mathbb{S}_0$ respectively by
$$
\mathbb{S}_1 =  \bigcup_{j=1}^J S_j
\quad\text{and}\quad
\mathbb{S}_0 = [-L/2,L/2] \setminus \left( \bigcup_{j=1}^J S_j \right).
$$
\end{defn}

It is known that kernel smoothing expands the signal support. This effect requires care as it increases the probability of obtaining false positives in the regions neighboring the signal \citep{Pacifico:2007}. The larger the bandwidth $\gamma$, the greater
the distortion. In particular, the support defined with respect to
$w_\gamma$ is larger than $S_j$, and we define it formally next.

\begin{figure}
\psset{unit=20mm}
\begin{center}
\begin{pspicture}(0,-1)(6,2.4)
\psline(0,0)(6,0)
\pscurve(2,0)(3,1.5)(4,0) 
\pscurve[linestyle=dashed](1,0)(3,1)(5,0)
\rput[lb](2.7,1.8){True signal}
\rput[lb](4.7,.6){Smoothed signal}
\psset{tbarsize=10pt}
\psline{|<->|}(2,-0.4)(4,-0.4)
\rput[lb](3,-0.35){$\mathbb{S}_1$}
\psline{|<->|}(1,-0.8)(5,-0.8)
\rput[lb](2.9,-0.75){$\mathbb{S}_{1, \gamma}$}
\psline{->|}(0,-0.4)(1.989,-0.4)
\rput[lb](1,-0.35){$\mathbb{S}_0$}
\psline{|<-}(4.01,-0.4)(6,-0.4)
\rput[lb](5,-0.35){$\mathbb{S}_0$}
\psline{->|}(0,-0.8)(0.989,-0.8)
\rput[lb](.5,-0.75){$\mathbb{S}_{0, \gamma}$}
\psline{|<-}(5.01,-0.8)(6,-0.8)
\rput[lb](5.5,-0.75){$\mathbb{S}_{0, \gamma}$}
\end{pspicture}
\caption{\label{fig:s0} True and smoothed signal.}
\end{center}
\end{figure}
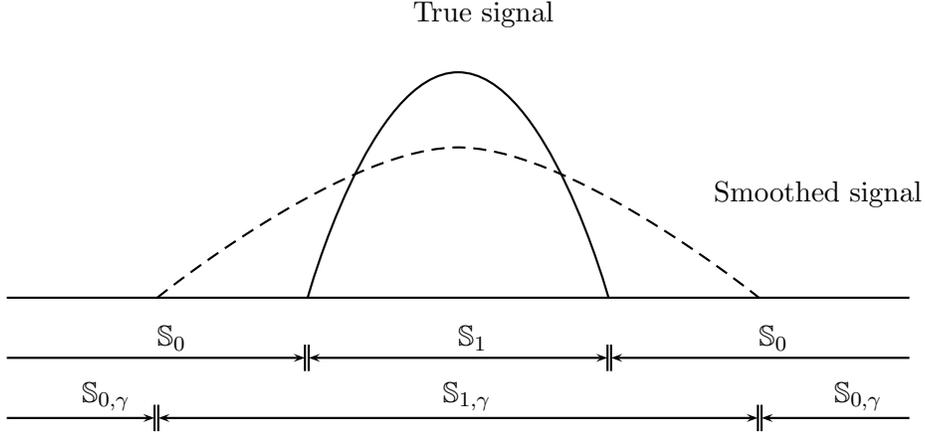

\begin{defn}
\label{defn:support_gamma}
Fix $\gamma > 0$ and define $ S_{j, \gamma} = \{t: h_\gamma(t-\tau_j)>0 \}$ where $h_\gamma(t)$ is given by \eqref{eq:hgamma}.
Define the expanded signal region $\mathbb{S}_{1, \gamma}$ and reduced null region $\mathbb{S}_{0, \gamma}$ respectively by
$$
\mathbb{S}_{1, \gamma} =  \bigcup_{j=1}^J S_{j, \gamma}
\quad\text{and}\quad
\mathbb{S}_{0, \gamma} = [-L/2,L/2] \setminus \left( \bigcup_{j=1}^J S_{j, \gamma}\right).
$$
\end{defn}

The above definitions are illustrated schematically in Figure \ref{fig:s0}. 
Some useful equalities are: (1) $\mathbb{S}_0 \cup \mathbb{S}_1 =
\mathbb{S}_{0, \gamma} \cup \mathbb{S}_{1, \gamma}$ and (2)
$\mathbb{S}_{1, \gamma}\setminus \mathbb{S}_1 =\mathbb{S}_{0}\setminus 
\mathbb{S}_{0, \gamma} $. We refer to the latter set, the difference between the expanded signal support and the true signal support, as the transition region.

We follow the classical view of error definitions. A significant local maximum that belongs to $\mathbb{S}_{0}$ is considered an error. A significant local maximum that belongs to $\mathbb{S}_{1}$ is considered correct. Moreover, if more than one local maximum occur within the same peak, all significant local maxima are considered as correct. However, all these local maxima are counted as one peak for the purposes of power, defined formally below in Section \ref{sec:power}. Thus power is not inflated. Nevertheless, this situation does not affect the theoretical results because under our asymptotic assumptions (Theorem \ref{thm:FDRcontrol}) each peak is represented by one local maximum with probability tending to 1. 

To simplify the  notation, we define the number of falsely
rejected local maxima using the threshold $u$ as:
$$
V(u) = \#\{t\in \tilde{T}\cap \mathbb{S}_0 : x_\gamma(t) >u\} ~\text{and}~ V_\gamma(u)
= \#\{t\in \tilde{T}\cap \mathbb{S}_{0, \gamma} : x_\gamma(t) >u\},
$$
where $V(u)$ counts false discoveries in the full null region and $V_\gamma(u)$ counts them only in the null region minus the transition region. Similarly, we define the number of correctly rejected hypotheses using the threshold $u$ as
$$
W(u) = \#\{t\in \tilde{T}\cap \mathbb{S}_1 : x_\gamma(t) >u\} ~\text{and}~ W_\gamma(u)
= \#\{t\in \tilde{T}\cap \mathbb{S}_{1, \gamma} : x_\gamma(t) >u\},
$$
where $W(u)$ counts discoveries in the true signal region and $W_\gamma(u)$ counts them in the signal region plus the transition region.  The total number of rejections using threshold $u$ is
$$
R(u) = V(u) + W(u) = \#\{t\in \tilde{T} : x_\gamma(t) >u\} .
$$
The number of tests where the null hypothesis is true is denoted by
$$
\tilde{m}_0 = V(-\infty) = \#\{ t\in  \tilde{T}\cap \mathbb{S}_0\},
$$
and the total number of tests is
$$
\tilde{m} = R(-\infty) = \#\{ t\in  \tilde{T}\}.
$$
The number of tests where the null hypothesis is false is denoted by
$\tilde{m}_1 =\tilde{m} - \tilde{m}_0$.

\begin{defn}
\label{defn:FWER}
Define the  family-wise error rate ($\FWER$) for any threshold $u$ as the probability that
there exists at least one local maximum of $x_\gamma(t)$ in the null region above the threshold $u$:
\begin{equation}
\label{eq:FWER}
\FWER(u) = \P\left\{ V(u)\ge 1 \right\} = \P\left\{ \tilde{T}\cap \mathbb{S}_0 \ne \emptyset ~\text{and}~
\max_{t \in \tilde{T}\cap \mathbb{S}_0} x_\gamma(t) > u
\right\}.
\end{equation}
\end{defn}

\begin{defn}
\label{defn:FDR}
Define the false discovery rate ($\FDR$) as the expected proportion of
falsely rejected hypotheses using a fixed threshold $u$:
\begin{equation}
\label{eq:FDR}
\FDR(u) = \E\left\{ \frac{V(u)}{R(u)\vee1} \right\}.
\end{equation}
\end{defn}

Note that if $\tilde{T}$ is empty then $\tilde{m}=0$ and $V(u)=0$ and so $\FWER(u) = \FDR(u)=0$. However, the probability of this event goes to zero as $L$ increases.
\subsection{Weak control of FWER}

In this subsection we assume the complete null hypothesis $\mu(t) = 0, \forall t\in [-L/2,L/2]$. Note that under this assumption $\tilde{m}=\tilde{m}_0$.

In Procedure \ref{alg:proc}, after Step 2 one has a list of $\tilde{m}$ local maxima. The Bonferroni procedure at level $\alpha$ defines a threshold of the form $\alpha/\tilde{m}$ and rejects all hypotheses whose p-values are below this threshold.
Note that in the usual Bonferroni procedure the number of tested hypotheses is  constant. In our case the number of tested hypotheses is random, therefore we consider two versions. We refer to the version in part (2) of Theorem \ref{thm:Bonferroni} as the Bonferroni procedure, while the version is part (1) can be viewed as the 'limit' of the Bonferroni procedure.

\begin{thm}
\label{thm:Bonferroni}
Assume the model of Section \ref{sec:model} and the procedure of Section \ref{sec:proc}. Let $\mu(t)=0, ~\forall t$. Fix $\alpha > 0$. 
\begin{enumerate}
\item Suppose the null hypothesis $\mH_0(t)$ is rejected at $t \in \tilde{T}$ if
\begin{equation}
\label{eq:threshold}
p_\gamma(t) < \frac{\alpha}{\E[\tilde{m}]}
\qquad \iff \qquad x_\gamma(t) >
u^*_{\Bon} = F_\gamma^{-1} \left(\frac{\alpha}{\E[\tilde{m}]} \right).
\end{equation}
Then $\FWER(u^*_{\Bon}) \le \alpha$.

\item Suppose the null hypothesis $\mH_0(t)$ is rejected at $t \in \tilde{T}$ by a Bonferroni procedure with $\tilde{m}$ tests, i.e. if
\begin{equation}
\label{eq:threshold-random}
p_\gamma(t) < \frac{\alpha}{\tilde{m}}
\qquad \iff \qquad x_\gamma(t) >
\tilde{u}_{\Bon} = F_\gamma^{-1} \left(\frac{\alpha}{\tilde{m}}\right).
\end{equation}
If $\tilde{m}=0$ we define $\alpha/\tilde{m}$ as infinity. 
Then $\limsup_{L \to \infty} \FWER(\tilde{u}_{\Bon}) \le \alpha$.
\end{enumerate}
\end{thm}

The proof of Theorem \ref{thm:Bonferroni} is given in Section \ref{App:3}. The first part of the theorem is not asymptotic and its proof is a direct consequence of the definition of p-values. The threshold  $u^*_{\Bon}$ in \eqref{eq:threshold} is deterministic. For example, in the Gaussian autocorrelation model this threshold can be computed by substituting \eqref{eq:Gaussian-ACVF}. The threshold $\tilde{u}_{\Bon}$ in \eqref{eq:threshold-random} depends on the random quantity $\tilde{m}$ and is  equivalent to applying the Bonferroni procedure on the random set of local maxima, $\tilde{T}$. The proof of the second part is based on the fact that by the weak law of large numbers  $\tilde{m}$ is close  to its expectation $\E[\tilde{m}]$ for large enough $L$.

Using \eqref{eq:pvapprox}, the thresholds in  Theorem \ref{thm:Bonferroni} may be approximated by 
\begin{equation*}
\begin{aligned}
u^*_{\Bon} &\approx \sigma_\gamma \phi^{-1}\left( \frac{\alpha}{L}
\sqrt{\frac{2\pi \sigma_\gamma^2}{\lambda_{2,\gamma}}} \right)
\approx \sigma_\gamma \sqrt{2 \log\left( \frac{L}{\alpha}
\sqrt{\frac{\lambda_{2,\gamma}}{2\pi \sigma_\gamma^2}}\right)} \\
\tilde{u}_{\Bon} &\approx \sigma_\gamma \phi^{-1}\left(
\frac{\alpha}{\tilde{m}_L} \sqrt{\frac{\lambda_{4,\gamma}
\sigma_\gamma^2}{2\pi \lambda^2_{2,\gamma}}} \right)
\approx \sigma_\gamma \sqrt{2 \log\left( 
\frac{\tilde{m}_L}{\alpha} \sqrt{\frac{2\pi
\lambda^2_{2,\gamma}}{\lambda_{4,\gamma} \sigma_\gamma^2}}\right)}
\end{aligned}
\end{equation*}
Notice that asymptotically both thresholds  have a similar form to the universal threshold of \citet{Donoho:1994,Donoho:1995}.

\subsection{Strong control of FWER}

Let us start with brief reminder about the assumptions of the true
model. We observe $y(t)$ as defined in \eqref{eq:signal+noise} in the
segment $T = [-L/2,L/2]$, where the signal contains $J_L$ peaks at locations
$\tau_j \in S_j$ for $1 \le j \le J_L$, where $S_j$ is the finite
support of the $j$th peak. Recall that $h_\gamma (t) =w_\gamma(t)*h(t)$ is the peak shape after smoothing. It achieves its supremum at a single point $\tau_{j, \gamma} \in S_j$ which is not necessarily the 
same as $\tau_j$ and has no other critical points. We define the signal after smoothing as
\begin{equation}
\label{eq:mu-smooth}
\mu_\gamma(t) = \sum_{j=1}^{J_L} a_j h_\gamma(t-\tau_{j, \gamma}).
\end{equation}
The next theorem describes the strong control of FWER for this 
model. 

\begin{thm}
\label{thm:strongFWER}
Assume the model of Section \ref{sec:model} and the procedure of Section \ref{sec:proc}. Fix $\alpha > 0$. 
\begin{enumerate}
\item Suppose the null hypothesis $\mH_0(t)$ is rejected at $t \in \tilde{T}$ if
\begin{equation}
\label{eq:threshold1}
p_\gamma(t) < \frac{\alpha}{\E[\tilde{m}]}
\qquad \iff \qquad x_\gamma(t) >
u^*_{\Bon} = F_\gamma^{-1} \left(\frac{\alpha}{\E[\tilde{m}]} \right).
\end{equation}
Then, for all $L$, $\limsup\FWER(u^*_{\Bon}) \le \alpha$, as $a_j \to\infty, \forall j$. 
\item Suppose the null hypothesis $\mH_0(t)$ is rejected at $t \in \tilde{T}$ by the Bonferroni procedure with $\tilde{m}$ tests, i.e. if
\begin{equation}
\label{eq:threshold1-random}
p_\gamma(t) < \frac{\alpha}{\tilde{m}}
\qquad \iff \qquad x_\gamma(t) >
\tilde{u}_{\Bon} = F_\gamma^{-1} \left(\frac{\alpha}{\tilde{m}}\right).
\end{equation}
If $\tilde{m}=0$ we define $\alpha/\tilde{m}$ as infinity.
Then $\limsup \FWER(\tilde{u}_{\Bon}) \le \alpha$, as $L \to \infty$
and $a_j \to\infty, \forall  j$, such that for all $j, ~ La_j\phi(Ka_j) \to
0$  for any constant $K$ .
\end{enumerate}
\end{thm}

Notice that the result in the first part of Theorem \ref{thm:strongFWER} is
asymptotic, in contrast to the result in the first part of Theorem \ref{thm:Bonferroni}. In addition, $u^*_{\Bon}$
defined in \eqref{eq:threshold1} cannot be computed, even assuming the 
Gaussian autocorrelation model, since the true signal region is unknown. In practice we recommend to apply the Bonferroni procedure defined in 
the second part of Theorem \ref{thm:strongFWER}, which can always be implemented and provides asymptotic strong control of FWER. The rates in Theorem \ref{thm:strongFWER} part (2) may be achieved, for instance, if $L$ increases exponentially with $a_j$.

The proof of Theorem \ref{thm:strongFWER} is given in Section \ref{App:4}.
There we show first that the expected number of local maxima in the transition region $\mathbb{S}_{1, \gamma} \setminus \mathbb{S}_1$ converges to 0 (Lemma \ref{lemma:transitionE}). Then we use arguments similar to those in the proof of Theorem \ref{thm:Bonferroni}.

\subsection{Control of FDR}
\label{sec:FDR-control}

In many applied problems the expected ratio of false discoveries is a
more appropriate error rate to control. The next theorem states that applying the BH procedure on the random set of local maxima controls the FDR asymptotically. 

\begin{thm}
\label{thm:FDRcontrol}
Assume the model of Section \ref{sec:model} and the procedure of Section \ref{sec:proc}.
Let $x_{\gamma}^{(i)}(t)$ be the $i$th ordered value of
$\tilde{T}$ in ascending order. Let
\begin{equation}
\label{eq:FDRproc}
\tilde{u}_{\BH} = u_k, \qquad k=\min_{1\le i\le \tilde{m}}\{i:x_{\gamma}^{(i)}(t) >u_i  \}
\end{equation}
be the threshold obtained by applying the BH procedure at the level $\alpha$ on the random set $\tilde{T}$ with $\tilde{m} \ne 0$  tests, where 
$$
 u_i = F_\gamma^{-1}\left(\frac{\tilde{m}-i+1}{\tilde{m}} \alpha\right), \qquad 1 \le i \le \tilde{m}
$$
are the constants corresponding to the BH procedure. Reject
$\tilde{m}- k+1$ hypotheses corresponding to $x_\gamma(t) > \tilde{u}_{\BH}$. If
such $k$ does not exist or $\tilde{m}=0$ reject nothing.
Assume that for $L \to \infty$ and $a_j \to \infty$ for all $j$, 
\begin{enumerate}
\item the number of peaks $J_L$ increases at the same rate as $L$ or
slower, i,e. $J_L/L \to A_1$, where $0 < A_1 < 1$.
\item $L \to \infty$ and $a_j \to\infty, \forall  j$ such that $ \forall j, ~ L\phi(Ka_j^\delta) \to
0$ for all $\delta>0$ and any constant $K$.
\end{enumerate}
Then, 
$$ 
\limsup \FDR(\tilde{u}_{\BH}) \le \alpha.
$$
\end{thm}
 
The rates of this theorem may be achieved if $L$ grows exponentially  with $a_j$. The proof is given in Section \ref{App:FDR}. First we show that: (1) a local maximum exists in a small neighbourhood of a true peak mode with probability tending to 1; (2) this local maximum is rejected for any fixed threshold with probability tending to 1 (Lemma \ref{lemma:unique-max}). The proof of the theorem is then based on the following arguments.
It is known that the  threshold of the BH procedure can be viewed as the largest solution of the equation $\alpha G(u) = F_\gamma(u)$, where $G(u)$ is the empirical right cumulative distribution function of $x_\gamma(t), ~t\in \tilde{T}$ \citep{Genovese:2002}. The ergodic assumption guarantees that $G(u)$ has a limit. Replacing $G(u)$ by its limit we find that the asymptotic  solution $u^*_{\BH}$ satisfies
\begin{equation}
\label{eq:FDRfixed-threshold}
F_\gamma(u^*_{\BH})=\frac{\alpha A_1}{A_1 + \E[\tilde{m}_{0, \gamma}; [0, 1]](1-\alpha)}.
\end{equation}
We show that under the conditions of the theorem, the threshold $u^*_{\BH}$ asymptotically controls the FDR below the desired level $\alpha$. Since $F_\gamma(\tilde{u}_{\BH}) \to F_\gamma(u^*_{\BH})$, the FDR level will be asymptotically controlled as well when using $\tilde{u}_{\BH}$ instead of $u^*_{\BH}$.

Notice that, in contrast to the Bonferroni procedure, where the deterministic threshold $u^*_{\Bon}$ in \eqref{eq:threshold1} grows unbounded with increasing $L$, the asymptotic threshold for the BH procedure $u^*_{\BH}$ in \eqref{eq:FDRfixed-threshold} is finite, and depends on the asymptotically fixed proportion of false null hypotheses, $A_1$.

\subsection{Power}
\label{sec:power}

We define the statistical power of Procedure \ref{alg:proc} as
the expected fraction of true discovered peaks:

\begin{equation}
\begin{aligned}
\label{eq:power}
\Power(u) &= \E \left[ \frac{1}{J_L} \sum_{j=1}^{J_L} 1\left(
\tilde{T} \cap S_j \ne \emptyset ~\text{and}~ \max_{\tilde{t}
\in \tilde{T} \cap S_j} x_\gamma(\tilde{t}) > u \right)\right] \\
&= \P\left\{
\tilde{T} \cap S_j \ne \emptyset ~\text{and}~ \max_{t 
\in \tilde{T} \cap S_j} x_\gamma(t) > u \right\}
\end{aligned}
\end{equation}

We discuss power in two ways. First, we show that the Bonferroni and BH procedures in \eqref{eq:threshold1-random} and \eqref{eq:FDRproc} respectively are consistent in the sense that the power tends to 1 for any fixed value of the smoothing parameter $\gamma$. Second, we discuss what value of the smoothing parameter $\gamma$  maximizes the power for finite samples.

\subsubsection{Power of the Bonferroni and BH procedures}

The consistency of both procedures is given in the next theorem.

\begin{thm}
\label{thm:power}
\hfill\par\noindent
\begin{enumerate}
\item Under the conditions of Theorem \ref{thm:strongFWER} the power of the Bonferroni procedure \eqref{eq:threshold1-random} converges in probability to 1.
\item Under the conditions of Theorem \ref{thm:FDRcontrol} the power
of the BH procedure \eqref{eq:FDRproc} converges in probability to 1.
\end{enumerate}
\end{thm}

The proof of Theorem \ref{thm:power} is based on the next lemma.
\begin{lemma}
\label{lemma:fwer-fdr-threshold}
\hfill\par\noindent
\begin{enumerate}
\item Under the conditions of Theorem \ref{thm:strongFWER},  $u^*_{\Bon}/[a_jh_\gamma(0)] \to 0$ in probability.
\item Under the conditions of Theorem \ref{thm:FDRcontrol}, $u^*_{\BH}/[a_jh_\gamma(0)] \to 0$ in probability. 
\end{enumerate}
\end{lemma}

As $L$ and $a_j$, $j=1,\dots,m$, go to infinity, the deterministic Bonferroni threshold \eqref{eq:threshold1} goes to infinity as well because $\E[\tilde{m}]$ goes to infinity. Lemma \ref{lemma:fwer-fdr-threshold} states that $a_jh_\gamma(0)$ goes to infinity faster than  $u^*_{\Bon}$ does. In contrast, the asymptotic BH threshold \eqref{eq:FDRfixed-threshold} does not depend on $L$ but only on the asymptotically fixed proportion of false null hypotheses, $A_1$. Therefore it does not increase with $L$ but is asymptotically constant.

Note that the statements in Lemma \ref{lemma:fwer-fdr-threshold} are about deterministic thresholds, while the statements of Theorem \ref{thm:power} are about random thresholds. As in the proofs of the previous theorems, we use the fact that the gap between the deterministic and random thresholds for both the Bonferroni and BH procedures goes to 0.

It is known that in general, if there exists a signal anywhere, the power of the BH procedure is larger then the power of Bonferroni procedure \citep{Benjamini:1995}. This is also true in our case with respect to $u^*_{\Bon}$ and $u^*_{\BH}$. To see this, note that for any fixed and large enough $L$ the thresholds can be approximated by

\begin{equation}
\label{eq:appBon}
F_\gamma(u^*_{\Bon})  =  \frac{\alpha}{\E[\tilde{m}]} \approx \frac{\alpha}{\tilde{m}_{1, \gamma}+\E[\tilde{m}_{0, \gamma}]}
\end{equation}
and
\begin{equation}
\label{eq:appBH}
F_\gamma(u^*_{\BH})  = \frac{\alpha A_1}{A_1 + \E[\tilde{m}_{0, \gamma}; [0, 1]](1-\alpha)} \approx \frac{\alpha \tilde{m}_{1, \gamma}}{\tilde{m}_{1, \gamma} + (1-\alpha)\E[\tilde{m}_{0, \gamma}]}.
\end{equation}
It immediately follows that if $\tilde{m}_{1, \gamma} \ge 1$, the threshold $u^*_{\Bon}$ is larger than the threshold $u^*_{\BH}$, promising a larger power for the BH procedure.

\subsubsection{Optimal choice of $\gamma$}
\label{sec:optimal-gamma}
Here we discuss the best choice of
$\gamma$, i.e. the value of $\gamma$ that maximizes the power \eqref{eq:power}. To maximize the probability of local maxima to exceed a given threshold $u$ under the true model is a difficult problem. Therefore, we turn to less formal discussion in this section. 

Lemma \ref{lemma:unique-max} in Section \ref{App:FDR} shows that for any single true peak with mode at $t=0$, a local maximum exists in a small neighbourhood $I_{j, \eps}$ of $t=0$ with probability tending to 1. The power may be approximated as
\begin{equation}
\begin{aligned}
\label{eq:approx-power}
\Power(u) &= \P\left\{
\tilde{T} \cap S_j \ne \emptyset ~\text{and}~ \max_{t 
\in \tilde{T} \cap S_j} x_\gamma(t) > u \right\} \\ \nonumber
&= \P\left\{
\tilde{T} \cap I_{j, \eps} \ne \emptyset ~\text{and}~ \max_{t 
\in \tilde{T} \cap I_{j, \eps}} x_\gamma(t) > u \right\} \\ \nonumber
& \qquad \qquad +
 \P\left\{
\tilde{T} \cap (S_j \setminus I_{j, \eps}) \ne \emptyset ~\text{and}~ \max_{t 
\in \tilde{T} \cap (S_j \setminus I_{j, \eps})} x_\gamma(t) > u \right\} \\ \nonumber
&\approx \P\left\{
\tilde{T} \cap I_{j, \eps} \ne \emptyset ~\text{and}~ \max_{t 
\in \tilde{T} \cap I_{j, \eps}} x_\gamma(t) > u \right\}  \\ \nonumber
&\approx \P\left\{\max_{t 
\in I_{j, \eps}} x_\gamma(t) > u ~\mid~ \text{a local maximum exists in} ~ I_{j, \eps} \right\}. \nonumber
\end{aligned}
\end{equation}
Therefore, heuristically, the optimal value of $\gamma$ should be close to the value of $\gamma$ that maximizes $\P(x_\gamma(0)>u)$. This value of   $\gamma$    is approximately given by:
\begin{equation}
\label{eq:max-gamma}
\begin{aligned}
\argmax_{\gamma} \P(x_\gamma(0)>u) &=
\argmax_{\gamma}\Phi\left(\frac{a_j h_\gamma(0) - u}{\sigma_\gamma} \right)
= \argmax_{\gamma} \left( \frac{h_\gamma(0)}{\sigma_\gamma} - \frac{u}{a_j\sigma_\gamma}\right) \\
&\approx \argmax_{\gamma}\frac{h_\gamma(0)}{\sigma_\gamma} =
\argmax_{\gamma} \frac{\int_{-\infty}^\infty w_\gamma(s) h_b(s)\,ds}
{\sigma \sqrt{\int_{-\infty}^{\infty} w^2_\gamma(s)\,ds}},
\end{aligned}
\end{equation}
where $\sigma$ and $\sigma_\gamma$ are the standard deviations of the observed and smoothed processes $y(t)$ and $x_\gamma(t)$, respectively. The approximation may be justified for both procedures by Lemma \ref{lemma:fwer-fdr-threshold}.
In the last expression in \eqref{eq:max-gamma}, the optimal value of $\gamma$ is that which makes $w_\gamma(t)$ closest to $h_b(t)$ in an $L^2$ sense. In particular, if $w(t) = h(t)$, then the optimal value of $\gamma$ is $b$. This result is similar to the well-known matched filter theorem for detecting a single signal peak of known shape at a fixed time $t$.

\begin{example}[{\bf Gaussian autocorrelation model}]
\label{ex:Gaussian-gamma-choice}
Suppose $h(t) = w(t) = \phi(t)$, the standard normal density, 
and assume that the autocovariance function of the noise
process $z(t)$ is proportional to a normal density as in Example \ref{ex:Gaussian-ACVF}.
In this case
\begin{equation}
\label{eq:optimal-gamma}
\argmax_\gamma \frac{h_\gamma(0)}{\sigma_\gamma} = \begin{cases}
\sqrt{b^2 - 2 \nu^2}, & \nu < b/\sqrt{2} \\
0, & \nu > b/\sqrt{2}
\end{cases}
\end{equation}
We show in the simulations below that the optimal $\gamma$ is indeed close to the one obtained by eq. \eqref{eq:optimal-gamma}.
\end{example}

\section{Simulation Studies}
\label{s:sim}
\subsection{Performance for finite space and SNR}
\label{sec:sim1}
A simulation study was carried out to investigate the FDR and FWER levels as well as the power of the proposed algorithm based on the Bonferroni and the BH procedures for a finite range $L$ and for a finite SNR. 
Our simulation study is based on the model described in Section \ref{sec:model} and Example \ref{ex:Gaussian-ACVF}, with 20 equally spaced and non overlapping peaks of the same height. The peak shape, $h(t)$, is the normal density with standard deviation $b=3$, truncated at $c=\pm 2b$. Peaks were centered at $\tau_j = 100j-50$, for $ j=1,\ldots ,20$ and  $a_j$ in \eqref{eq:mu} for all $j$ is chosen to be 10 and 15 in two different scenarios, giving a moderate and a strong SNR respectively. The noise is a stationary zero-mean Gaussian process as defined in Example \ref{ex:Gaussian-ACVF} with $\sigma=1$ and $\nu = 0,1$ and 2.   We sampled the time axis at $L=2000$ equally spaced points. Figure \ref{fig:Data} presents a fragment of the simulated data. Note that the height of the peaks is $a_j/\sqrt{2\pi b^2}$.

\begin{figure}[t]
\begin{center}
     \includegraphics[width=12cm]{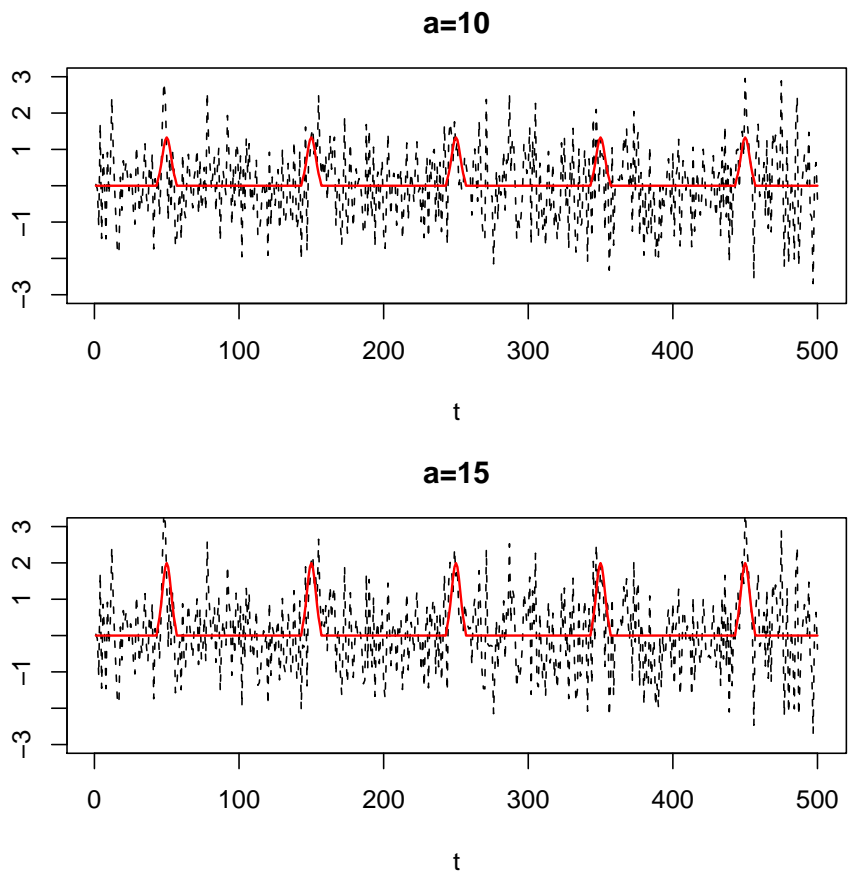}
\end{center}
\caption{ \label{fig:Data} A fragment of the simulated data.}
\end{figure}

Following Procedure \ref{sec:proc}, we constructed the process $x_\gamma(t) = y(t)*w_\gamma(t)$, where $w_\gamma(t)$ is a Gaussian kernel and $\gamma$ was equally spaced between 1 to 6.5 in steps of 0.5. For each configuration we computed the power and the error level based on 10,000 replications. The FWER level was computed as the proportion of replications for which at least one false discovery occurred. The FDR level was computed as the mean proportion of false discoveries out of the total number of discoveries. The power is presented as the mean number of truly discovered peaks out of the 20 peaks.  

\begin{figure}[t]
\begin{center}
   \includegraphics[width=10cm]{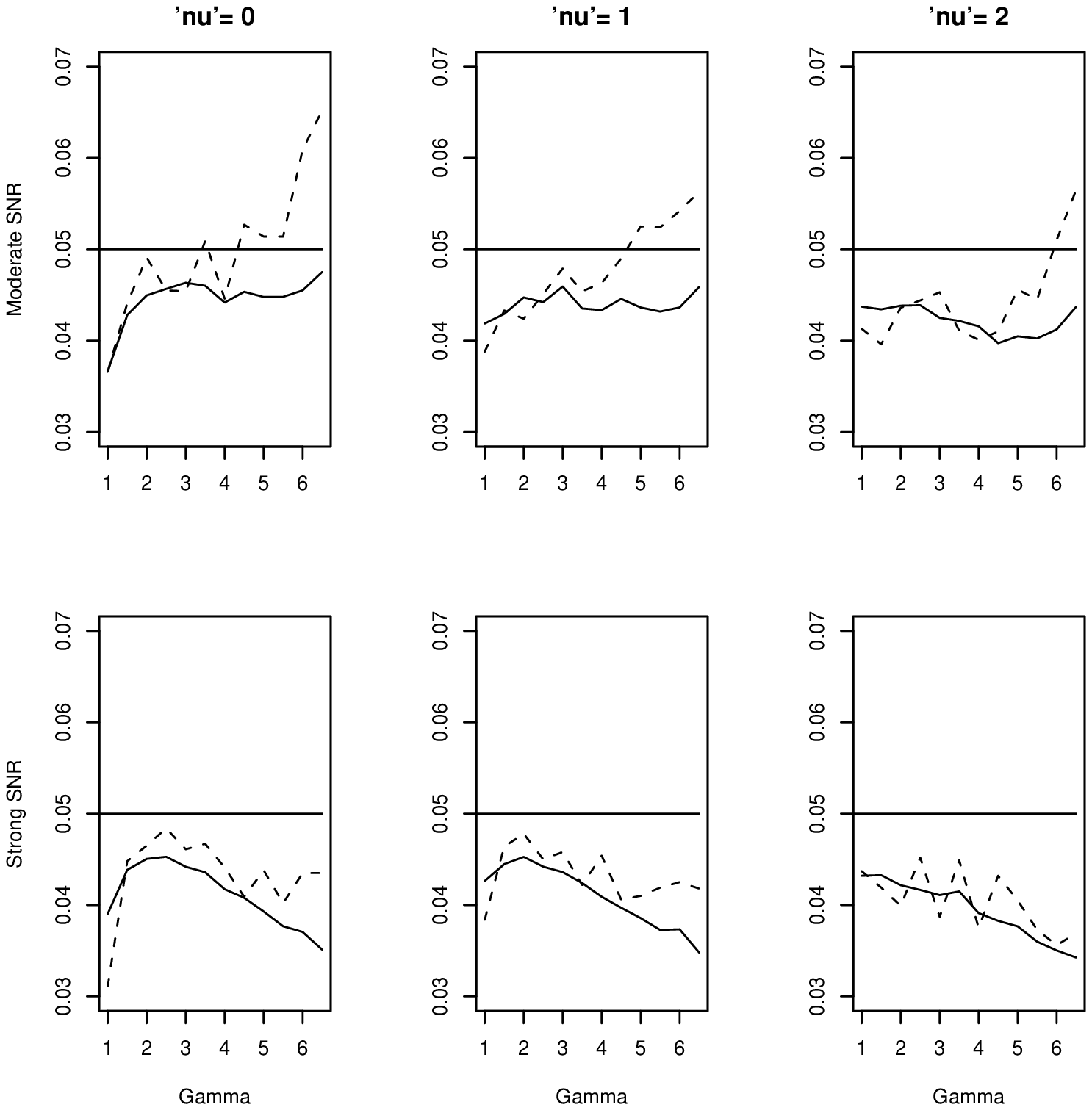}
 \caption{ \label{fig:err} FDR level of the BH procedure (solid), FWER level of the Bonferroni procedure (dashed).}
 \end{center}
 \end{figure}

Figure \ref{fig:err} presents the FWER and FDR levels of the Bonferroni and the BH procedures for all the studied configurations. The nominal level for both procedures was set to 0.05. For relatively large values of $\gamma$ the Bonferroni procedure may exceed the prespecified error level. This happens due to the broadening of the signal: many of the rejected local maxima correspond to real peaks, but due to large value of $\gamma$, the modes are shifted away from the true mode locations to the transition region, and the discoveries are no longer considered as correct. FDR is less affected because FDR is an expected ratio, so larger $\gamma$ also results in fewer discoveries in the transition region. The effect of the broadening of the signal is less for both procedures for larger SNR. Note that for any finite SNR one can choose large enough $\gamma$, so the phenomenon of the broadening of the signal will lead to exceedance of the error rate. Recall, however, that asymptotically there should be no local maxima in the transition region.
There is no strong dependency on SNR and on the noise autocorrelation parameter $\nu$.

\begin{figure}[t]
\begin{center}
   \includegraphics[width=10cm]{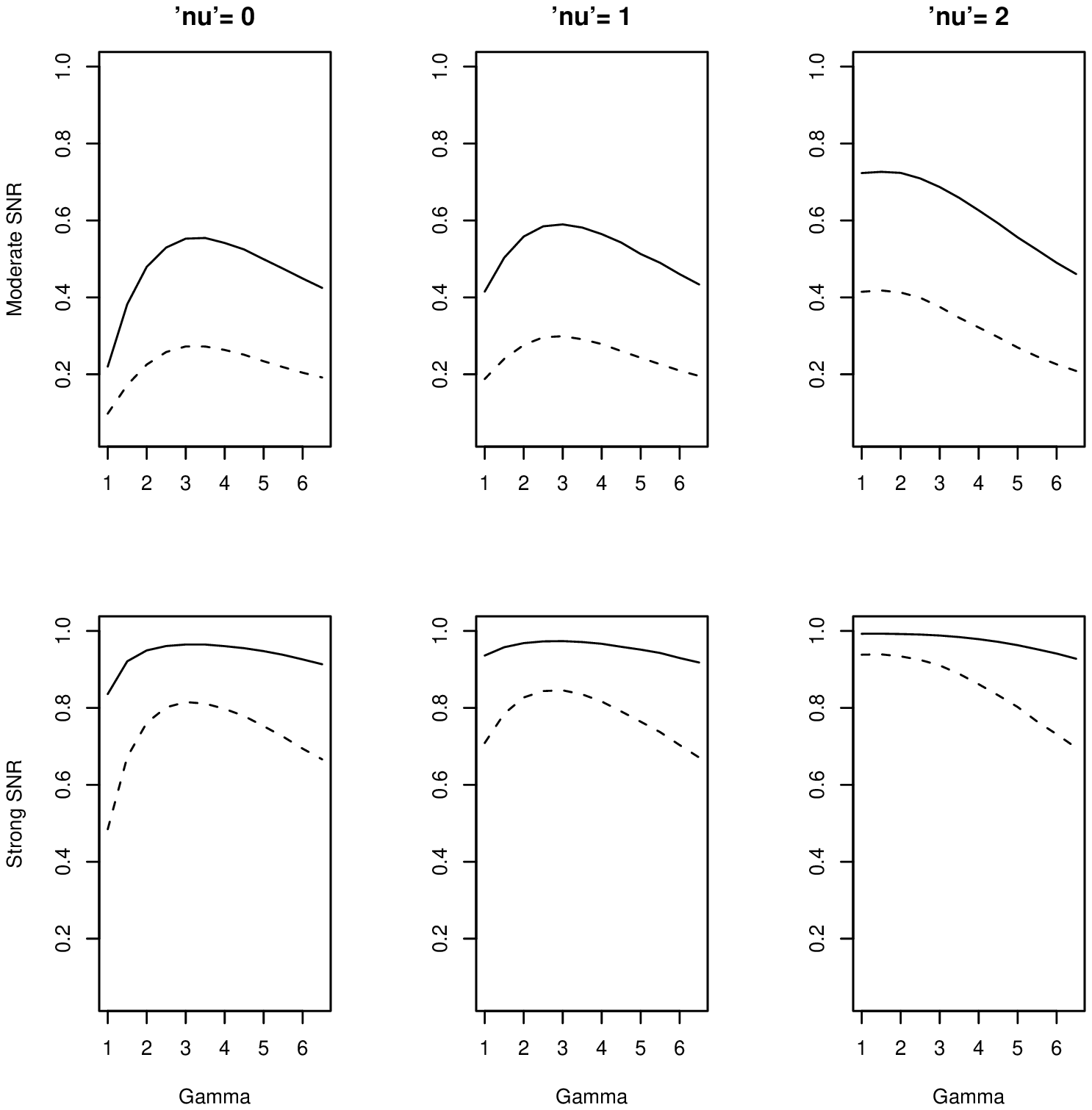}
 \caption{  \label{fig:paperPWR} Power of the BH (solid) and Bonferroni (dashed) procedures. }
\end{center}
\end{figure}

Figure \ref{fig:paperPWR} presents the power of the Bonferroni and the BH procedures, respectively, for all studied configurations. As expected, (1) the power of the BH procedure is greater than that of the Bonferroni procedure for all studied configurations; (2) the power of both procedures is greater for larger SNR. For large SNR the power of the BH procedure is almost a flat function of $\gamma$. This means that for large SNR the selection of $\gamma$ is not very important as long as it is near the signal width $b$. 
 
Figure \ref{fig:paperPWR}  shows that there is one value of $\gamma$ in each configuration for which the power is maximized. The optimal value is usually around $b=3$, but it depends on the parameter $\nu$. The larger $\nu$, the smaller the value of the optimal $\gamma$ as expected from \eqref{eq:optimal-gamma}. Note that the error level is always controlled if one selects the $\gamma$ value close to the optimal one. To get more precise results about the optimal value of $\gamma$ we performed an additional simulation study where $\gamma$ takes values in the range 1 to 3.5 in steps of 0.1. The empirical optimal values of $\gamma$, as well as the optimal $\gamma$  from equation \eqref{eq:optimal-gamma}, are summarized in Table \ref{t:gamma}. 

\begin{figure}[t]
\begin{center}
   \includegraphics[width=10cm]{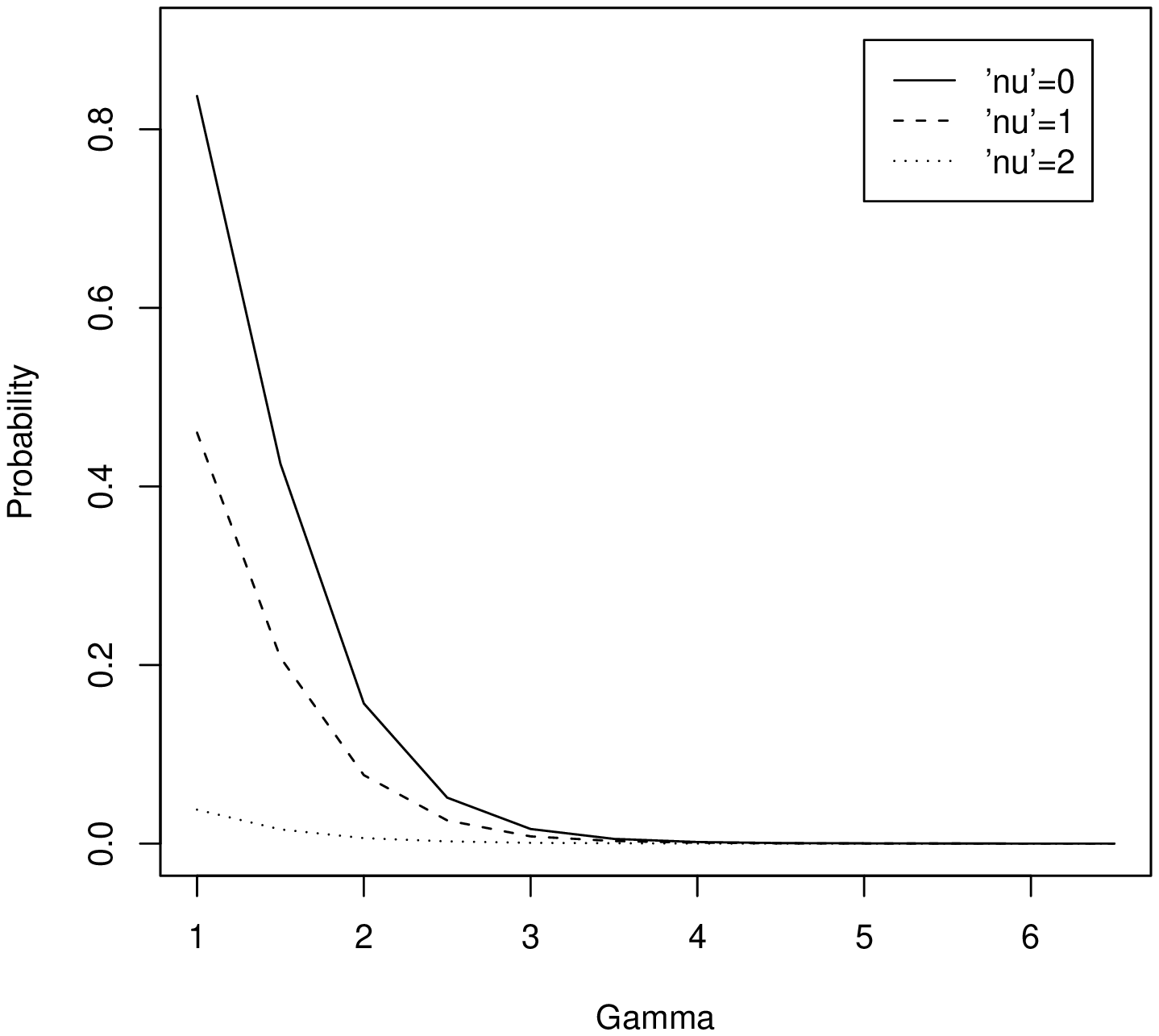}
\caption{   \label{fig:prob1} Probability of obtaining more than one local maximum within the same true peak for moderate SNR ($a=10$).}
\end{center}
\end{figure}

\begin{table}[t]
\begin{center}
\begin{tabular}{l||l|l|l||l|l|l||} \hline
\multicolumn{1}{c}{} & \multicolumn{3}{c}{$a_j=10$} &
\multicolumn{3}{c}{$a_j=15$ } \\
\hline
 $\nu = $ & 0 & 1 & 2 & 0 & 1 & 2  \\ \hline
 equation \eqref{eq:optimal-gamma} & 3.0 & 2.8 & 1.0 &3.0 & 2.8 & 1.0 \\ \hline
 Bon	& 3.3	& 2.8&	1.3 &	2.9	& 2.9	& 1.3 \\ \hline
 BH	& 3.2 &	3.1	& 1.3	& 3.4	& 3.0	& 1.2  \\ \hline
\end{tabular} 
\caption{\label{t:gamma} The `optimal' value of $\gamma$.}
\end{center}
\end{table}

As explained in Section \ref{sec:FDR-control}, each peak is represented by no more than one local maximum with probability tending to 1 as the SNR increases. In our simulation we computed the probability of obtaining more than one local maximum within the same peak in the finite setting. Figure \ref{fig:prob1} shows that this probability is a decreasing function of $\gamma$, and a decreasing function of the SNR. Note that near the 'optimal' value of $\gamma$ this probability is negligible.

\subsection{Overlapping peaks}
In the second part of our simulation we test the performance of our algorithm for a non-sparse signal. In the previous simulation, the distance between two adjacent peak locations $\tau_j$, which we denote by $D$, was set to 100. We now test how the distance between the peaks $D$ affects the power and the error level of both procedures. 

\begin{figure}[t]
\begin{center}
   \includegraphics[width=10cm]{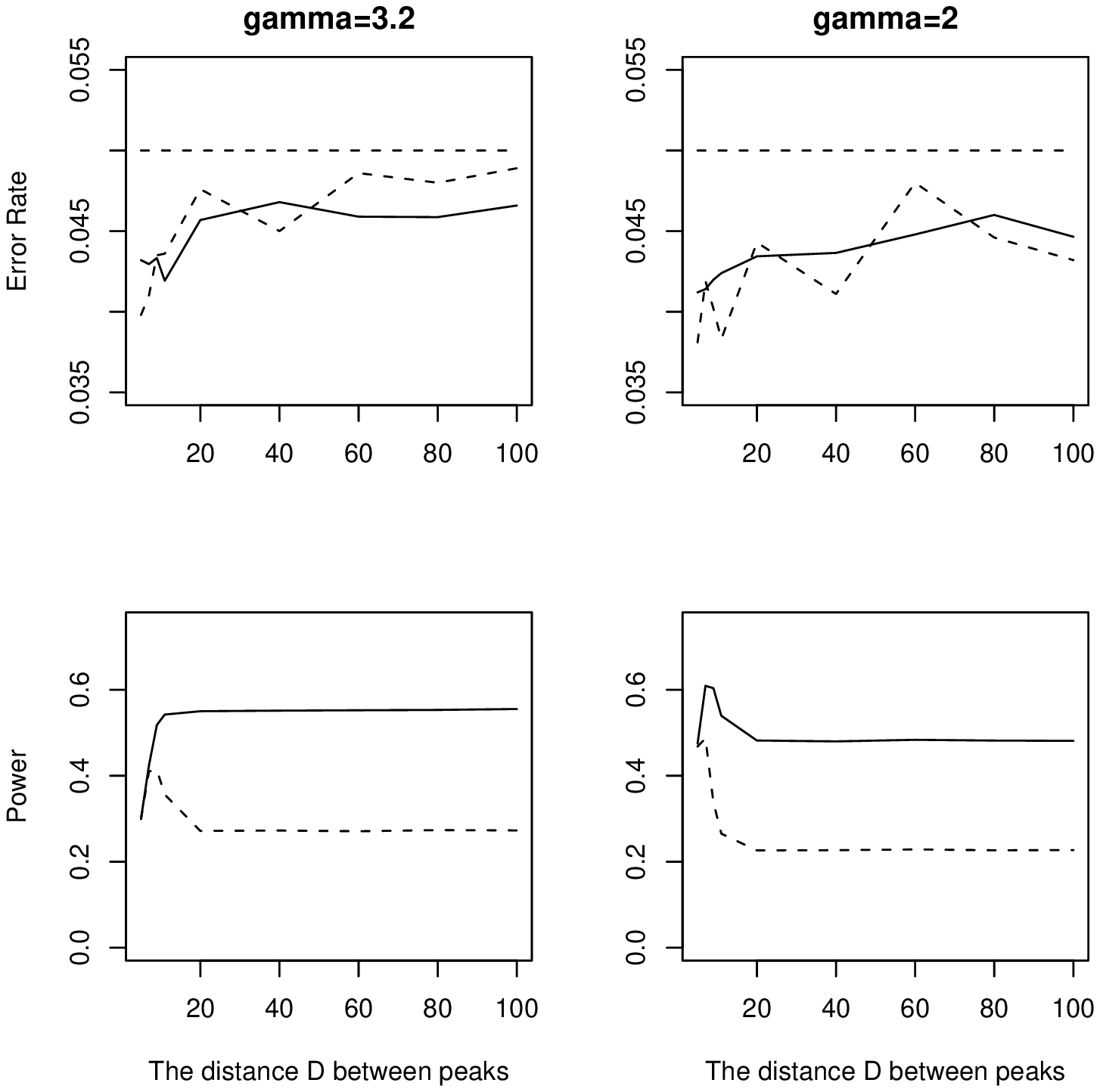}
\caption{ \label{fig:D} Effect of distance between peaks $D$ on the BH procedure (solid) and the Bonferroni procedure (dashed). }
\end{center}
\end{figure}

For the case where $a_j=10,~ \nu=0$ and $\gamma=3.2$ (which is optimal for the Bonferroni procedure) and $\gamma=2.0$, we reduce $D$ and $L$ while keeping fixed the proportion between the signal region, $\mathbb{S}_1$ and the total tested region, $\mathbb{S}$. The first row in Figure \ref{fig:D} presents the error levels of the Bonferroni and BH procedures for different values of $D$. Recall that the support of each peak has a length of 13 ($\tau \pm 2b$), so peaks start to overlap for $D<4b=12$. It can be seen that the error levels of both procedures are controlled regardless of the value of $D$.

Overlapping of peaks requires care of the definition of power. In this section the power was defined as the number of truly rejected peaks, where `truly' means that the rejected local maximum belongs to the peak support. Since any rejected local maximum must represent just one peak, we  divide the overlapping 
part of the support to two equal parts and add  the left half to the rejection region of the peak from the left and the right half  to the rejection region of the peak from the right. Thus, for any single peak the rejection region consists of the non-overlapping part of the support (in the middle of the peak) and half of each one of the overlapping parts. This makes the rejection region smaller than the original support. 
 
The power of both procedures almost does not change when $D$ is reduced until there is overlap in the supports; see the second row in Figure \ref{fig:D}. Because peaks in overlapping positions are combined, the number of local maxima in the $\mathbb{S}_1$ region decreases. However, these local maxima become more significant and are usually rejected by both the Bonferroni and the BH procedures. In the configuration presented here  for $\gamma=3.2$, for moderate overlapping (around 50\%)  20 peaks are represented on average by 8 local maxima, while for non-overlapping supports in the same configuration 20 peaks are presented by 19.5 local maxima.  In the latter case the mean number of rejections by the Bonferroni procedure is around 5 and the mean number of rejections by the BH procedure is around 11. This explains why for moderately overlapping supports the power of the BH procedure decreases and the power of the Bonferroni procedure increases. When peaks overlap by more than $50\%$ the power of both procedures decreases. For $\gamma=2$ the behaviour of power is the same. The power is relatively low before overlapping; therefore it increases for moderate overlapping and falls down for large overlapping. The last panel in Figure \ref{fig:D} shows that when peaks are close or have little overlap, small values of $\gamma$ perform better.

\section{Data Example}
\label{s:data}
\subsection{Data description and analysis}
\label{s:data1}
Our data consists of 60 seconds of recordings from an electrode attempting to capture the activity of a single type of neuron in a salamander brain. The recorded signal was digitized at a sampling frequency of 10 KHz, resulting in a sequence of 600,000 measurements equally spaced over time. Data of these kind and over much longer time periods are routinely collected in neuroscience experiments \citep{Baccus:2002}.

This is a situation where the model of Section \ref{sec:model} can be applied directly. Since all peaks of interest correspond to action potentials from the same type of cell, it is reasonable to assume that they all have the same shape. Their intensities, however, vary according to the distances of the cells to the recording electrode. The depolarization (positive) component of the action potential is unimodal. The noise is a mixture of electrical noise and recording of remote cells, exhibiting roughly a large-scale stationary behaviour. Favourably, the SNR is high. For illustration, the (smoothed) data is depicted in Figure \ref{fig:spike}.

\begin{figure}
\begin{center}
   \includegraphics[width=11cm,height=9cm,totalheight=10cm]{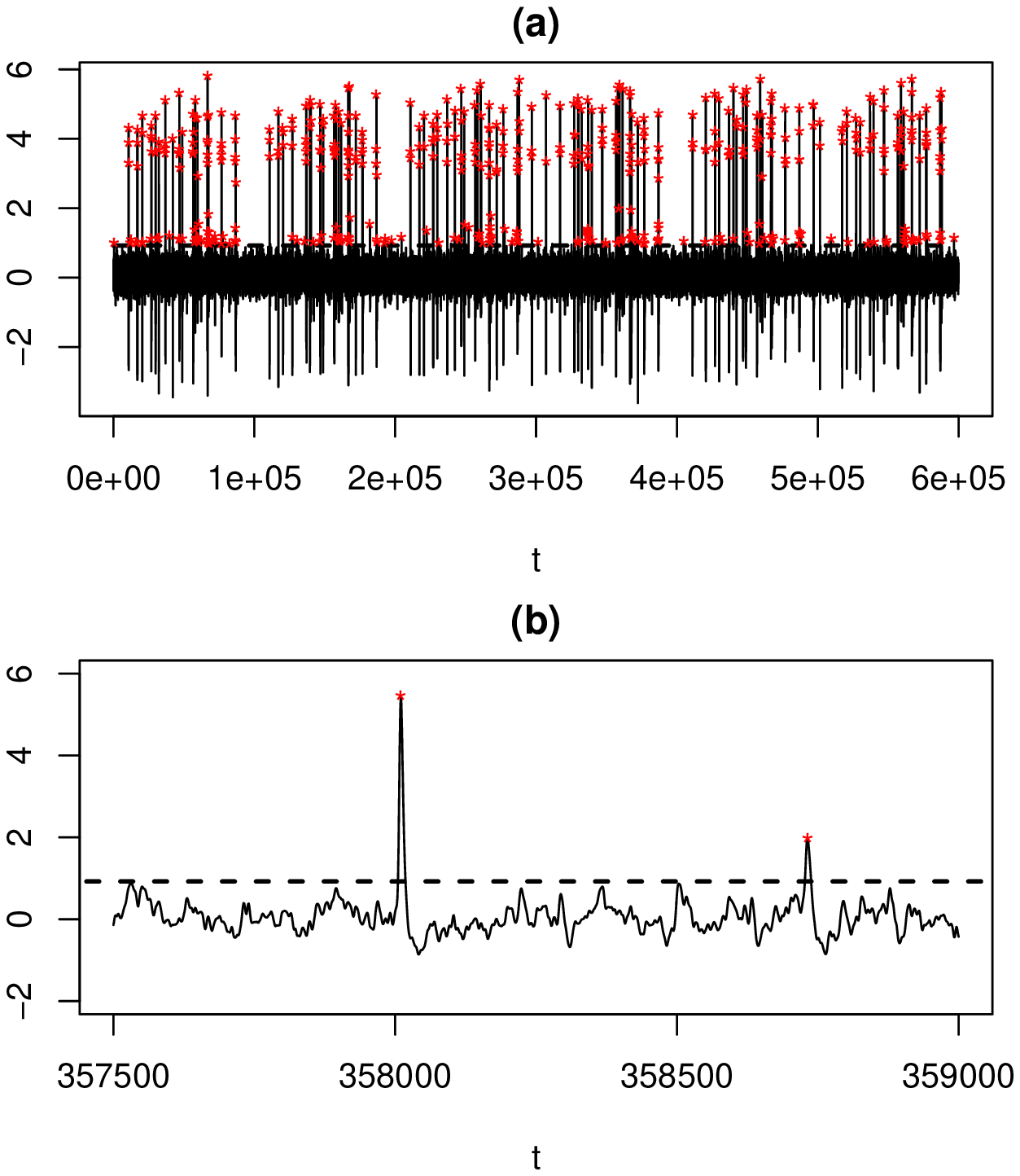}
 \caption{\label{fig:spike} (a) The neural spike data after smoothing with a Gaussian kernel of standard deviation 1.5. (b) Zoom in. In both panels, the stars indicate the detected peaks according to the BH procedure. The dashed line is the BH threshold for rejecting the null hypothesis.}
\end{center}
\end{figure}

Procedure \ref{alg:proc} was implemented as follows. For Step 1, we used a Gaussian kernel as in Example \ref{ex:Gaussian-ACVF}. From the results of Section \ref{s:sim}, the choice of kernel width $\gamma$ is not crucial but it is better if it roughly matches the width of the signal peaks. Since all peaks are assumed to have the same shape, we selected a few `obvious' peaks and estimated their width. Figure \ref{fig:peak-shape} shows that a scaled normal density with standard deviation of 1.5 approximates the peaks' shape reasonably well. 

\begin{figure}
\begin{center}
   \includegraphics[width=8cm]{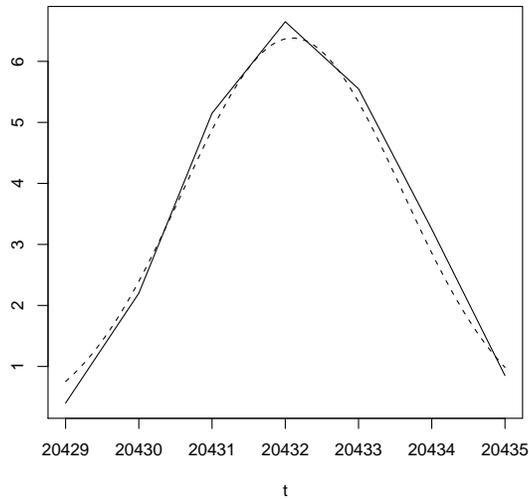}
 \caption{\label{fig:peak-shape} One `obvious' peak (solid) superimposed on a scaled normal density with standard deviation 1.5 (dashed).}
\end{center}
\end{figure}

In Step 2, $\tilde{m}=54,452$ local maxima were found. For the heights of these local maxima, the corresponding p-values in Step 3 were computed according to formula \eqref{eq:Gaussian-ACVF}, plugging in estimates of the spectral moment parameters $\sigma^2$, $\lambda_2$ and $\lambda_4$ as described below.

For Step 4, we applied the BH procedure with $\tilde{m}=54,452$ and level $q=0.01$, leading to 464 rejections of the null hypothesis. Figure \ref{fig:spike} shows the smoothed data and the BH threshold. While many of the peaks could have been found by simple inspection, thanks to the high SNR, the algorithm detected many other weaker peaks. Figure \ref{fig:spike}(b) shows a zoomed segment of the data, in which two quite different peaks were detected, one strong and one weak.

\subsection{Estimation of spectral moments}
Here we describe how to estimate the noise parameters $\sigma^2$, $\lambda_2$ and $\lambda_4$ using a simple method motivated by \eqref{eq:moments}. It is known that the variance of an ergodic process can be consistently  estimated by the sample variance. This suggests estimating $\sigma^2$, $\lambda_2$ and $\lambda_4$ respectively by the sample variance of the process, the difference process (as an approximation to the derivative), and the difference squared process (as an approximation to the second derivative). However, the smoothed observed process $x_\gamma(t)$ is expected to contain some signal and not just noise, biasing the sample variance estimators. To reduce the bias, we replace the standard sample variance by a nonparametric estimator that is less sensitive to extreme values, specifically
\begin{equation}
\label{eq:mad}
\begin{aligned}
\hat{\sigma}^2 &= \mad^2\left[x_\gamma(t)\right] = \med^2\left\{ \left| x_\gamma(t_i) - \med\left[x_\gamma(t)\right]\right|\right\}, \\
\hat{\lambda}_2 &=\mad^2\left[\Delta x_\gamma(t)\right], \\
\hat{\lambda}_4 &=\mad^2\left[\Delta^2 x_\gamma(t)\right],
\end{aligned}
\end{equation}
where $\Delta x(t)$ for any discretely sampled vector $x(t)=\{x(t_i)\}_{i=1}^L$ is the sequence of differences $\{[x(t_{i+1})-x(t_i)]/(t_{i+1}-t_i)\}_{i=1}^{L-1}$ . These estimators are expected to perform well if the signal is sparse in the sense that it occupies only a small portion of the data.

To evaluate the accuracy of the proposed estimators, we performed a brief simulation in which we compared their performance with the performance of three additional estimation methods. The first method estimates the spectral moments by standard sample variances. The second method is based on the fact that for a stationary zero-mean Gaussian process $z(t) \in \mathcal{C}^2$ with autocovariance function (ACF) $c(s) = \E[z(t) z(t+s)]$, the spectral moments satisfy $\sigma^2 = c(0)$, $\lambda_2 = -\ddot{c}(0)$ and $\lambda_4 = c^{(4)}(0)$. If $\hat{c}(s)$ is the empirical ACF of the noise process, we fit a polynomial regression $\hat{c}(s) = \beta_0 + \beta_2 s^2 + \beta_4 s^4 +\eps$ in the neighbourhood of $s=0$, where even powers of $s$ suffice because the ACF is symmetric around zero. Computing the derivatives of the polynomial at $s=0$ leads to the estimators $\hat{\sigma}^2 = \hat{\beta}_0$, $\hat{\lambda}_2 = -2\hat{\beta}_2$ and $\hat{\lambda}_4 = 24\hat{\beta}_4$. The third method estimates $\sigma^2$ as in \eqref{eq:mad} and $\lambda_2$ by the `crossing' estimator suggested by \citet{Lindgren:1974},
$$
\hat{\lambda}_2 = \hat{\sigma}^2 \frac{2\pi}{T}\frac{1}{3}\left( N_0(T) + \exp{\frac{u^2}{2}}N_u(T) + \exp{\frac{u^2}{2}}N_{-u}(T) \right),
$$
where $u=2\hat{\sigma}/3$ and $N_u(T)$ is the number of upcrossings of the process $x_\gamma(t)$ of the level $u$ in $[0, T]$. The estimation of $\lambda_4$ is done in the same way from the process of differences $\Delta x_\gamma(t)$.

In the simulation, two discrete sequences of length 10,000 were generated. The first one contained only white Gaussian noise and the second one contained white Gaussian noise plus 17 equally shaped peaks with heights all equal to 2, corresponding to high SNR. Mimicking the data, both sequences were smoothed with a Gaussian kernel of standard deviation 1.5. The exact moments $\sigma^2$, $\lambda_2$ and $\lambda_4$ were computed via \eqref{eq:Gaussian-moments}. The estimated moments by all the methods were averaged over 2000 replications. Table \ref{tbl:moments-est} summarizes the results.  

\begin{table}
\begin{center}
\text{Noise only} \\
\begin{tabular}{|l|l|l|l|l|l|} \hline
& Formula \eqref{eq:Gaussian-moments} & MAD & Var & ACF & Lindgren\\ \hline
$\sigma^2$ & 0.188 & 0.188 (0.007) & 0.188 (0.005) & 0.160 (0.005) & 0.188 (0.007)\\ \hline
$\lambda_2$ & 0.042 & 0.040 (0.001) & 0.040 (0.001) & 0.014 (0.0004) & 0.032 (0.001)\\ \hline
$\lambda_4$ & 0.009 & 0.023 (0.001) & 0.023 (0.001) & 0.002 (0.0001)& 0.016 (0.001) \\ \hline
\end{tabular} 
\\
\vspace{3 mm}
\text{Noise plus sparse signal} \\
\begin{tabular}{|l|l|l|l|l|l|} \hline
& Formula \eqref{eq:Gaussian-moments} & MAD  & Var  & ACF  & Lindgren\\ \hline
$\sigma^2$ & 0.188 & 0.193 (0.007) & 0.201 (0.005) & 0.171 (0.005) & 0.193 (0.007)\\ \hline
$\lambda_2$ & 0.042 & 0.040 (0.001) & 0.041 (0.001) & 0.015 (0.005) & 0.032 (0.001)\\ \hline
$\lambda_4$ & 0.009 & 0.024 (0.001) & 0.024 (0.001) & 0.002 (0.0001)& 0.016 (0.001) \\ \hline
\end{tabular} 
\caption{\label{tbl:moments-est} Simulation results: exact and estimated spectral moments by four different methods. Standard deviations are in parentheses.}
\end{center}
\end{table}

Returning to the data, after subtracting the overall mean of 0.0537, the obtained estimates via \eqref{eq:mad} were $\hat{\sigma}^2 = 0.0587$, $\hat{\lambda}_2 = 0.0019$ and  $\hat{\lambda}_4 = 0.0006$. P-values were then computed for the mean-subtracted data via \eqref{eq:distr} using these estimated moments, leading to the results described earlier in Section \ref{s:data1}.

\section{Discussion}
\label{s:discussion}

In this paper, we have used the heights of local maxima after smoothing as test statistics for identifying unimodal peaks in the presence of Gaussian stationary noise. It was shown that the procedure provides strong control of FWER and FDR asymptotically as both the SNR and length of the sequence tend to infinity, with the length of the sequence allowed to grow exponentially faster than the SNR. Simulations showed that the algorithm is powerful and that a matched filter principle applies where the optimal smoothing bandwidth is close to the width of the peaks to be detected.

The most critical assumptions for the theoretical results presented are that the noise process is stationary ergodic Gaussian and that the signal peaks have equal shape and are unimodal with compact support. The Gaussianity assumption was chosen because it enabled using a closed formula for computing the p-values associated with the heights of local maxima. For non-Gaussian noise, p-values could be computed via simulation and we expect that the error control properties should be preserved, although this does not follow directly from the proofs presented here.

The assumption of compact support for the signal peaks is necessary for the concept of true and false detection to be well defined. The unimodality assumption makes local maxima good representatives of true peaks. This is formally true asymptotically for high SNR, as the probability that a true peak is represented by one and only observed local maximum tends to one. Besides its practical interpretation, this property is also helpful technically in the proof of FDR control. We do not find the assumption of asymptotically high SNR restrictive in the sense that the search space is allowed to grow exponentially faster. Another way of seeing this is that the SNR need only grow logarithmically with respect to the search space. It was shown in the simulations that moderate SNR suffices for good performance.

The assumption that the peaks have equal shape is technically convenient as it allows reducing the asymptotic analysis of many peaks to the analysis of any single one of them. It also simplifies the concept of an optimal bandwidth, as it is the same for all peaks. This assumption is also a realistic one in many practical situations, such as the neural spikes example presented.

Notice that there is no assumption of sparsity of peaks in the theoretical results. The simulations showed that the error rates and power do not suffer much even if the peaks have some overlap. Sparsity is not needed as long as the spectral moments of the noise process are known. However, sparsity is needed so that these moments can be estimated from data, as suggested in the data analysis section.

A technical issue to consider in practice is that the theory was developed for continuous processes, while in simulations and real data the observations are obtained in a regular discrete grid. The theoretical results carry through approximately if the grid is fine enough, but break down if the smoothing parameter $\gamma$ is less than the grid spacing. As it is well known in signal processing, sampling before smoothing, inevitable in practice in Step 1 of the procedure, introduces aliasing. Convolution of a sampled discrete sequence with a sampled discrete kernel is not equivalent to sampling the convolution of a continuous process with a continuous kernel, and the approximation gets worse as the smoothing parameter $\gamma$ gets close to 1. For this reason, we did not include values of $\gamma$ between 0 to 1 in our simulation study, and we generally do not recommend to smooth using a value of $\gamma$ that is too close to the grid spacing, since for that case the theory is not valid. 

From a broad perspective, we see the methods presented in this paper not only as a solution to the peak detection problem but as an extension of multiple testing paradigms to temporal and spatial domains. While FWER methods for random fields and have been well established, particularly in neuroimaging \citep{Worsley:2004}, similar extensions of FDR methods have proven to be more difficult \citet{Pacifico:2004,Pacifico:2007}. Other related approaches include testing for a spatial signal in the wavelet transform domain \citep{Shen:2002} and FDR for pre-defined spatial clusters \citep{Heller:2007}. Standard FDR methods can be applied in a discretized spatial domain \citep{Genovese:2002} but are inherently designed for discrete units and ignore the spatial structure of the data. Instead, it has been argued by \citet{Chumbley:2009} that in the case of smooth spatial signals, inference should be about topological features, such as cluster volume or peak height. These authors and others \citep{Zhang:2009} have focused on cluster volume. Our worked has focused on peak height.

Potential extensions of this work include exploration of the role of sparsity in the estimation of the noise parameters, as well as adapting the procedure to situations where the observed process is not Gaussian, where the peaks do not have a constant width, or where the domain is two- or three-dimensional, as in medical image analysis.

\section{Proofs}
\label{s:proofs}

\subsection{Gaussian autocorrelation model}
\label{App:2}
\begin{lemma}
\label{lemma:Gaussian}
Let $w_\nu(t) = (1/\nu) \phi(t/\nu)$, where $\phi(t)$ is the standard
normal density.
\begin{enumerate}
\item For $\gamma, \nu > 0$,
$$
w_\gamma(t) * w_\nu(t) = w_\xi(t), \qquad
\text{with} ~~ \xi = \sqrt{\gamma^2 + \nu^2}.
$$
\item Let $\dot{w}_\xi(t)$ and $\ddot{w}_\xi(t)$ denote the first and
second derivatives of $w_\xi(t)$ with respect to $t$. Then
$$
\int_{-\infty}^\infty \big[w_\xi(t)\big]^2\,dt = 
\frac{1}{2\sqrt{\pi}\xi}, \qquad
\int_{-\infty}^\infty \big[\dot{w}_\xi(t)\big]^2\,dt = 
\frac{1}{4\sqrt{\pi}\xi^3}, \qquad
\int_{-\infty}^\infty \big[\ddot{w}_\xi(t)\big]^2\,dt = 
\frac{3}{8\sqrt{\pi}\xi^5}.
$$
\end{enumerate}
\end{lemma}

\begin{proof}
\hfill\par\noindent
\begin{enumerate}
\item
Let $X \sim N(0,\gamma^2), Y \sim N(0,\nu^2)$ with respective
densities $w_\gamma(t)$ and $w_\nu(t)$. Then $X + Y \sim
N(0,\gamma^2+\nu^2)$ with density $w_\gamma(t) * w_\nu(t) =
w_\xi(t)$, $\xi = \sqrt{\gamma^2 + \nu^2}$.

\item
Let $w_\xi^{(k)}(t)$ denote the $k$-th derivative of $w_\xi(t)$
and let $H_k(t)$ denote the $k$-th Hermite polynomial. Then
$$
\int_{-\infty}^\infty \big[w_\xi^{(k)}(t)\big]^2\,dt = 
\int_{-\infty}^\infty \left[\frac{(-1)^k}{\xi^k}
H_k\left(\frac{t}{\xi}\right) \frac{1}{\xi}
\phi\left(\frac{t}{\xi}\right)\right]^2\,dt =
\frac{1}{\xi^{2k+2}} \int_{-\infty}^\infty
H_k^2\left(\frac{t}{\xi}\right)
\phi^2\left(\frac{t}{\xi}\right)\,dt.
$$
But
$$
\phi^2\left(\frac{t}{\xi}\right) = \frac{1}{\sqrt{2\pi}}
\phi\left(\frac{\sqrt{2} t}{\xi}\right).
$$
Thus replacing in the integral and making the change of variable
$x = \sqrt{2} t/\xi$, we obtain
$$
\int_{-\infty}^\infty \big[w_\xi^{(k)}(t)\big]^2\,dt = 
\frac{1}{2\sqrt{\pi} \xi^{2k+1}}
\int_{-\infty}^\infty H_k^2\left(\frac{x}{\sqrt{2}}\right) \phi(x)\,dx.
$$
The results of the lemma are obtained by setting in particular $k = 0$ 
with $H_0(x) = 1$, $k = 1$ with $H_1(x) = x$, and $k = 2$ with $H_2(x)
= x^2 - 1$.
\end{enumerate}
\end{proof}

\noindent
\textbf{Derivations for Example \ref{ex:Gaussian-ACVF} (Gaussian autocorrelation model)}
\hfill\par\noindent
Under the complete null hypothesis, we can write
\begin{equation}
\label{eq:x-null}
x_\gamma(t) = w_\gamma(t) * z(t) =
w_\gamma(t) * \sigma \int_{-\infty}^\infty w_\nu(t-s)\,dB(s) =
\sigma \int_{-\infty}^\infty w_\xi(t-s)\,dB(s)
\end{equation}
with $\xi = \sqrt{\gamma^2 + \nu^2}$, where we have used Lemma
\ref{lemma:Gaussian} part (1). By Lemma \ref{lemma:Gaussian} part (2),
\begin{equation*}
\begin{aligned}
\sigma_\gamma^2 &= \E[x_\gamma^2(t)]
= \sigma^2 \int_{-\infty}^\infty w_\xi^2(t-s) \,ds
= \frac{\sigma^2}{2\sqrt{\pi}\xi} \\
\lambda_2 &= \E[\dot{x}_\gamma^2(t)]
= \sigma^2 \int_{-\infty}^\infty \dot{w}_\xi^2(t-s) \,ds
= \frac{\sigma^2}{4\sqrt{\pi}\xi^3} \\
\lambda_4 &= \E[\ddot{x}_\gamma^2(t)]
= \sigma^2 \int_{-\infty}^\infty \ddot{w}_\xi^2(t-s) \,ds
= \frac{3\sigma^2 }{8\sqrt{\pi}\xi^5}.
\end{aligned}
\end{equation*}
\subsection{Proof of Theorem \ref{thm:Bonferroni} (Weak control of FWER)}
\label{App:3}

\begin{lemma}
\label{lemma:threshold}
Let $u^*_{\Bon}$ and $\tilde{u}_{\Bon}$ be the thresholds defined in
\eqref{eq:threshold} and \eqref{eq:threshold-random}, respectively.
Then $|\tilde{u}_{\Bon} - u^*_{\Bon}| \to 0$ in probability as $L \to \infty$.
\end{lemma}

\begin{proof}
Recall that $\tilde{m}_0$ is the number of local maxima belonging to 
the set $\mathbb{S}_0$ in the segment $[-L/2, L/2]$. Denote by $\tilde{m}_0[0,1]$ the number of local maxima
belonging to the set $\mathbb{S}_0$ in the unit.
Using the notation of Lemma \ref{lemma:m_u}, we have that by ergodicity,
$$
\left|\frac{\tilde{m}_0}{L} - E[\tilde{m}_0[0,1]]\right|
\to 0
$$
in probability as $L \to \infty$, where
$E\big[\tilde{m}_0[0,1]\big]$ does not depend on $L$. Since $\log(\cdot)$
is continuous, the continuous mapping theorem gives that
$$
\left|\log \frac{\tilde{m}_0}{L} - \log E[\tilde{m}_0[0,1]]\right|
\to 0 \qquad \Longrightarrow \qquad
\left|\log \frac{\tilde{m}_0}{\alpha} - \log
\frac{E[\tilde{m}_0]}{\alpha} \right| \to 0,
$$
where we have used the additive property of the logarithm.

Define now the monotone increasing function $\psi_\gamma: \R^+ \to \R$ by
$\psi_\gamma(x) = F^{-1}_\gamma(1-e^{-x})$, $x > 0$, where $\kappa
> 0$ is a constant.
The function $\psi_\gamma(x)$ is Lipschitz continuous for all $x > 1$ because 
its derivative $d\psi_\gamma(x)/dx =
e^{-x}/\dot{F}_\gamma[\psi_\gamma(x)]$ is bounded for all $x >
1$. Hence,
$$
\left|\psi_\gamma\left(\log \frac{\tilde{m}_0}{\alpha}\right) -
\psi_\gamma\left(\log \frac{E[\tilde{m}_0]}{\alpha}\right)
\right| \to 0,
$$
implying that
$$
\left|F^{-1}_\gamma\left(1-\frac{\alpha}{\tilde{m}_0}\right) -
F^{-1}_\gamma\left(1-\frac{\alpha}{E[\tilde{m}_0]}\right)\right|
\to 0
$$
as $L \to \infty$. Multiplying by $\sigma_\gamma$ gives the result.
\end{proof}

\begin{proof}[{\bf Proof of Theorem \ref{thm:Bonferroni}}]
\hfill\par\noindent
\begin{enumerate}
\item 
By Proposition \ref{thm:p-values} and Lemma \ref{lemma:m_u},
$$
\begin{aligned}
\FWER(u^*_{\Bon}) = \P\{V(u^*_{\Bon}) \ge 1\}
&\le \E[V(u^*_{\Bon})] = L \E[\tilde{m}(u^*_{\Bon}; [0,1])] \\
&= L \E[\tilde{m}(-\infty; [0,1])]
\frac{\E[\tilde{m}(u^*_{\Bon}; [0, 1])]}{\E[\tilde{m}(-\infty; [0,1])]}
= \E[\tilde{m}_0] F_\gamma(u^*_{\Bon})
\end{aligned}
$$
by the Lemma \ref{lemma:m_u}, part (2). Setting
$F_\gamma(u^*_{\Bon}) = \alpha/\E[\tilde{m}_0]$ gives that 
$\FWER(u^*_{\Bon}) \le \alpha$.

\item
Write
$$
\FWER(\tilde{u}_{\Bon}) = \P\left\{ \tilde{T} \ne \emptyset ~\text{and}~
\max_{t \in \tilde{T}} x_\gamma(t) > u^*_{\Bon} +
(\tilde{u}_{\Bon} - u^*_{\Bon}) \right\}
$$
and apply the fact that for any two random variables $X$, $Y$ and any two constants $c$, $\eps$:
$$
P(X > c + \eps) - P(|Y| > \eps) \le P(X > Y + c) \le
P(X > c - \eps) + P(|Y| > \eps).
$$

Taking $X = \max_{t \in \tilde{T}} x_\gamma(t)$,
$Y = \tilde{u}_{\Bon} - u^*_{\Bon}$ and $c = u^*_{\Bon}$,
$$
\FWER(\tilde{u}_{\Bon}) \le \P\left\{ \tilde{T} \ne \emptyset ~\text{and}~
\max_{t \in \tilde{T}} x_\gamma(t) > u^*_{\Bon} - \eps 
\right\} +
\P\left\{\tilde{T} \ne \emptyset ~\text{and}~ |\tilde{u}_{\Bon} - u^*_{\Bon}|
> \eps\right\}
$$
The second summand goes to 0 in probability as $L \to \infty$ by Lemma 
\ref{lemma:threshold}. For the first summand, we follow a similar
argument to the proof of part (1):
$$
\P\left\{ \tilde{T} \ne \emptyset ~\text{and}~
\max_{t \in \tilde{T}} x_\gamma(t) > u^*_{\Bon} - \eps 
\right\} \le \E[\tilde{m}_0]  F_\gamma(u^*_{\Bon} - \eps) 
= \alpha \frac{ F_\gamma(u^*_{\Bon} - \eps)}{ F_\gamma(u^*_{\Bon})}
$$
but the last fraction goes to 1 as $L \to \infty$.
\end{enumerate}
\end{proof}

\subsection{ Proof of Theorem \ref{thm:strongFWER} (Strong control of FWER)}
\label{App:4}

\begin{lemma}
\label{lemma:transitionE}
Assume the model of Section \ref{sec:model} and the procedure of Section \ref{sec:proc}. 
Then as $L \to \infty$ and $a_j \to\infty, \forall  j$ such that $ \forall j, ~ La_j\phi(Ka_j) \to
0$, for any constant $K$, 
\begin{enumerate}
\item
The expected number of local maxima in the transition
region $\mathbb{S}_{1, \gamma} \setminus \mathbb{S}_1 = \mathbb{S}_0\setminus \mathbb{S}_{0, \gamma}$ tends to 0:
$$
\E\left[ \#\{t \in \tilde{T}\cap\left(\mathbb{S}_{1, \gamma}\setminus \mathbb{S}_1
\right) \} \right] = \E\left[ \#\{t \in
\tilde{T}\cap\left(\mathbb{S}_0\setminus \mathbb{S}_{0, \gamma}
\right) \} \right] \to 0.
$$
\item
The probability of obtaining a local maximum in the transition region tends to 0:
$$
\P\left(\# \{t\in \tilde{T}\cap ( \mathbb{S}_0\setminus  \mathbb{S}_{0, \gamma})\}\ge 1 \right) \to 0
$$
\end{enumerate}
\end{lemma}

\begin{proof}
\hfill\par\noindent
\begin{enumerate}
\item
The expected number of local maxima in any set $T \subset [-L/2, L/2]$ can be computed by the Kac-Rice formula:
\begin{eqnarray}
\label{eq:expectedmax}
&  & \E\left[\#\{t \in T:  \dot{x}_\gamma(t) = 0 , \ddot{x}_\gamma(t) < 0
\}  \right] \nonumber \\
& = & \int_{T} p\left(\dot{x}_\gamma(t) = 0\right) \int_{-\infty}^{0}|y|p\left(
\ddot{x}_\gamma(t)=y\right)\,dy\,dt \nonumber  \\
& = & \int_{T } p\left(\dot{z}_\gamma(t) = - \dot{\mu}_\gamma(t) \right) \int_{0}^{\infty} y\, p\left(
\ddot{z}_\gamma(t)= - y - \ddot{\mu}_\gamma(t)\right)\,dy\,dt,
\end{eqnarray}
where $p(\cdot)$ denotes probability density. Recall that  $\ddot{z}_\gamma(t) \sim N (0, \lambda_{4, \gamma})$. The inner integral in \eqref{eq:expectedmax} is
\begin{eqnarray}
\label{eq:inner-int}
&  & \int_{0}^{\infty} y\,p\left(
\ddot{z}_\gamma(t)= - y - \ddot{\mu}_\gamma(t)\right)\,dy =
\int_{0}^{\infty}y\frac{1}{\sqrt{\lambda_{4, \gamma}}}\phi\left(\frac
{-y-\ddot{\mu}_\gamma(t)}{\sqrt{\lambda_{4, \gamma}}} \right)\,dy \nonumber \\
& = & \int_{\ddot{\mu}_\gamma(t)/\sqrt{\lambda_{4, \gamma}}}^{\infty}\left[\sqrt{\lambda_{4, \gamma}} z - \ddot{\mu}_\gamma(t)\right]\phi(z)\,dz \nonumber \\
& = &  - \ddot{\mu}_\gamma(t) + \sqrt{\lambda_{4, \gamma}}\phi\left(\frac{\ddot{\mu}_\gamma(t)}{\sqrt{\lambda_{4, \gamma}}}
\right) + \ddot{\mu}_\gamma(t)\Phi\left(\frac{\ddot{\mu}_\gamma(t)}{\sqrt{\lambda_{4, \gamma}}}\right).
\end{eqnarray}
The set $T$ can be decomposed as the union  $T = T^+ \cup T^-$, such that $ \ddot{\mu}_\gamma(t)>0, ~ \forall t\in T^+$ and $\ddot{\mu}_\gamma(t)< 0, ~ \forall t\in T^-$. Using the inequality
\begin{equation}
\label{eq:Gaus-ineq}
\frac{x}{1+x^2} \phi(x) < 1 - \Phi(x) < \frac{\phi(x)}{x}, \qquad x>0,
\end{equation}
for $ t\in T^+$, the expression in \eqref{eq:inner-int} is bounded by
\begin{eqnarray}
\label{eq:T+}
0 \le  \ddot{\mu}_\gamma(t)\left[\frac{ \sqrt{\lambda_{4, \gamma}}}{
\ddot{\mu}_\gamma(t)} \phi\left(\frac{\ddot{\mu}_\gamma(t)}{\sqrt{\lambda_{4, \gamma}}}
\right) + \Phi\left(\frac{\ddot{\mu}_\gamma(t)}{\sqrt{\lambda_{4, \gamma}}}
\right) - 1\right] \le
\phi\left(\frac{\ddot{\mu}_\gamma(t)}{\sqrt{\lambda_{4, \gamma}}}\right)\frac{\sqrt{\lambda_{4, \gamma}}}{1+\ddot{\mu}_\gamma^2(t)/\lambda_{4, \gamma}}.
\end{eqnarray}
For $ t\in T^-$ the sum of the last two terms of the expression in \eqref{eq:inner-int}  is bounded by
\begin{eqnarray}
\label{eq:T-}
0 & \le &  \sqrt{\lambda_{4, \gamma}}\phi\left(\frac{\ddot{\mu}_\gamma(t)}{\sqrt{\lambda_{4, \gamma}}}
\right) + \ddot{\mu}_\gamma(t)\Phi\left(\frac{\ddot{\mu}_\gamma(t)}{\sqrt{\lambda_{4, \gamma}}}
\right)\nonumber \\
&  = & -\ddot{\mu}_\gamma(t)\left[
\frac{\sqrt{\lambda_{4, \gamma}}}{-\ddot{\mu}_\gamma(t)}\phi\left(\frac{-\ddot{\mu}_\gamma(t)}{\sqrt{\lambda_{4, \gamma}}} 
\right) - 1 + \Phi\left(\frac{-\ddot{\mu}_\gamma(t)}{\sqrt{\lambda_{4, \gamma}}}
\right)\right] \nonumber \\ 
& \le &
\phi\left(\frac{-\ddot{\mu}_\gamma(t)}{\sqrt{\lambda_{4, \gamma}}}\right)\frac{\sqrt{\lambda_{4, \gamma}}}{1+\ddot{\mu}_\gamma^2(t)/\lambda_{4, \gamma}}.
\end{eqnarray}
Taking $T = \mathbb{S}_{1, \gamma}\setminus \mathbb{S}_1 = \cup_{j=1}^{J_L} S_{j,
\gamma}\setminus S_j$ and decomposing  each transition region as
$S_{j,\gamma}\setminus S_j = (S_{j,\gamma} \setminus S_j)^{+} \cup (S_{j,
\gamma}\setminus S_j)^-$ and replacing in \eqref{eq:expectedmax} we get

\begin{eqnarray}
\label{eq:int1}
0 & \le & \E\left[\#\{t \in \mathbb{S}_{1,\gamma}\setminus
\mathbb{S}_{1}:  \dot{x}_\gamma(t) = 0 , \ddot{x}_\gamma(t) < 0  \}  \right]
\nonumber \\
& = & J_L  \Biggl\{ \int_{(S_{j,\gamma}\setminus S_j)^+ } p\left(
\dot{z}_\gamma(t) = - \dot{\mu}_\gamma(t) \right)\left[ -
\ddot{\mu}_\gamma(t) +
\sqrt{\lambda_{4, \gamma}}\phi\left(\frac{\ddot{\mu}_\gamma(t)}{\sqrt{\lambda_{4, \gamma}}} 
\right) + \ddot{\mu}_\gamma(t)\Phi\left(\frac{\ddot{\mu}_\gamma(t)}{\sqrt{\lambda_{4, \gamma}}}
\right)\right]dt   \nonumber \\
& + &  \int_{(S_{j,\gamma}\setminus S_{j})^- } p\left(
\dot{z}_\gamma(t) = - \dot{\mu}_\gamma(t) \right)\Biggl[
\sqrt{\lambda_{4, \gamma}}\phi\left(\frac{\ddot{\mu}_\gamma(t)}{\sqrt{\lambda_{4, \gamma}}} 
\right) + \ddot{\mu}_\gamma(t)\Phi\left(\frac{\ddot{\mu}_\gamma(t)}{\sqrt{\lambda_{4, \gamma}}}
\right)\Biggr]dt  \nonumber \\
& + &  \int_{(S_{j,\gamma}\setminus S_{j})^- } p\left(
\dot{z}_\gamma(t) = - \dot{\mu}_\gamma(t) \right) \left[-
\ddot{\mu}_\gamma(t)\right]dt \Biggr\}.
\end{eqnarray}
The last integral in \eqref{eq:int1} equals 
\begin{eqnarray}
\label{eq:int4}
0 & \le &   \int_{(S_{j,\gamma}\setminus S_{j})^- }
p\left(\dot{z}_\gamma(t) = - \dot{\mu}_\gamma(t) \right) 
\left[-\ddot{\mu}_\gamma(t)\right]dt \nonumber \\
& = & \int_{(S_{j,\gamma}\setminus S_{j})^- } \frac{1}{\sqrt{2\pi\lambda_{2, \gamma}}}
\exp{\left( -\frac{\dot{\mu}_\gamma(t)^2}{2\lambda_{2, \gamma}}\right)}
\left(-\ddot{\mu}_\gamma(t)\right)dt. 
\end{eqnarray}

Since $\ddot{\mu}_\gamma(t)$ is piecewise continuous, we can express the set $(S_{j,\gamma}\setminus S_{j})^- $ as the finite union of closed intervals $[c_i, d_i]$, some to the left of $\tau_{j, \gamma}$ and some to the right. Note that $ \dot{\mu}_\gamma(t) $ does not change  sign for all $t \in [c_i, d_i]$ and it is a decreasing function, since $\ddot{\mu}_\gamma(t)<0$. Therefore, $\dot{\mu}_\gamma(c_i)> \dot{\mu}_\gamma(d_i)$. The integral  in \eqref{eq:int4} can be expressed as the finite sum of integrals of the form
\begin{eqnarray}
\label{eq:int5}
& & \int_{c_i }^{d_i} \frac{1}{\sqrt{2\pi\lambda_{2, \gamma}}}
\exp{\left( -\frac{\dot{\mu}_\gamma(t)^2}{2\lambda_{2, \gamma}}\right)}
\left[-\ddot{\mu}_\gamma(t)\right]dt  = \int_{\dot{\mu}_\gamma(d_i)/\sqrt{\lambda_{2, \gamma}} }^{\dot{\mu}_\gamma(c_i)/\sqrt{\lambda_{2, \gamma}}}\phi(z)\,dz \nonumber \\
& = & \Phi\left(\frac{\dot{\mu}_\gamma(d_i)}{\sqrt{\lambda_{2, \gamma}}} \right) - \Phi\left(\frac{\dot{\mu}_\gamma(c_i)}{\sqrt{\lambda_{2, \gamma}}} \right) \le \phi\left(\frac{\dot{\mu}_\gamma(g_i)}{\sqrt{\lambda_{2, \gamma}}} \right)
\frac{a_j[\dot{h}_\gamma(c_i)-\dot{h}_\gamma(d_i)]}{\sqrt{\lambda_{2, \gamma}} }
\end{eqnarray}
where $g_i = d_i$ if the interval is to the left of $\tau_{j, \gamma}$ and $g_i = c_i$ if the interval is to the right of $\tau_{j, \gamma}$.
Using the above result, the integral in \eqref{eq:int4} is bounded by:
\begin{eqnarray}
\label{eq:int4bound}
& & \int_{(S_{j,\gamma}\setminus S_{j})^- } p\left(
\dot{z}_\gamma(t) = - \dot{\mu}_\gamma(t) \right) \left[-
\ddot{\mu}_\gamma(t)\right]dt \le \phi\left(\frac{\dot{\mu}_\gamma(g^*)}{\sqrt{\lambda_{2, \gamma}}} \right)\frac{\Delta \dot{h}_\gamma}{\sqrt{\lambda_{2, \gamma}}}a_j
\end{eqnarray}
where $g^* = \arg\max\phi\left(\dot{\mu}_\gamma(g_i)/\sqrt{\lambda_{2, \gamma}} \right)$ and $\Delta \dot{h}_\gamma = \dot{h}_\gamma(c_i)-\dot{h}_\gamma(d_i)$ is the maximal difference.
Note that the bounds in \eqref{eq:T+} and \eqref{eq:T-} are the same. Plugging the bounds \eqref{eq:T+}, \eqref{eq:T-} and \eqref{eq:int4bound} in equation \eqref{eq:int1} we get
\begin{eqnarray}
\label{eq:int6}
0 & \le & \E\left[\#\{t \in \mathbb{S}_{1,\gamma}\setminus
\mathbb{S}_{1}:  \dot{x}_\gamma(t) = 0 , \ddot{x}_\gamma(t) < 0  \}  \right]
\nonumber \\
& \le & J_L  \Biggl\{ \int_{(S_{j,\gamma}\setminus S_j)} p\left(
\dot{z}_\gamma(t) = - \dot{\mu}_\gamma(t) \right)\left[ \phi\left(\frac{\ddot{\mu}_\gamma(t)}{\sqrt{\lambda_{4, \gamma}}}\right)\frac{\sqrt{\lambda_{4, \gamma}}}{1+\ddot{\mu}_\gamma^2(t)/\lambda_{4, \gamma}} \right]dt   \nonumber \\
& + &  \phi\left(\dot{\mu}_\gamma(g^*)/\sqrt{\lambda_{2, \gamma}} \right)\frac{\Delta \dot{h}_\gamma}{\sqrt{\lambda_{2, \gamma}}}a_j \Biggr\}.
\end{eqnarray}

For $t \in S_{j, \gamma} \setminus S_j, \forall j$, $\mu_\gamma(t) =
a_jh_\gamma(t-\tau_{j, \gamma})$. $\dot{h}_\gamma(t)$ is bounded away from 0, in particular $|\dot{h}_\gamma(t)|$ is bounded in $ S_{j, \gamma}
\setminus S_j$. Let $t^*$ be the point in  $S_{j, \gamma} \setminus S_j$ 
where $|\dot{h}_\gamma(t)|$ is minimal. Continuing \eqref{eq:int5}, the expected number of local maxima in $\mathbb{S}_{1, \gamma} \setminus \mathbb{S}_1$ is bounded above by
\begin{eqnarray}
&  & J_L \Bigg[\frac{1}{\sqrt{\lambda_{2,\gamma}}} \phi\left( \frac{a_j \dot{h}_\gamma(t^*)}{\sqrt{\lambda_{2,
\gamma}}}\right)\int_{S_{j,\gamma}\setminus S_j
}\phi\left(\frac{\ddot{\mu}_\gamma(t)}{\sqrt{\lambda_4}}\right)\frac{\sqrt{\lambda_{4,\gamma}}}{1+\ddot{\mu}_\gamma^2(t)/\lambda_4}dt 
+ \phi\left( \frac{a_j \dot{h}_\gamma(g^*)}{\sqrt{\lambda_{2,
\gamma}}}\right)\frac{\Delta \dot{h}_\gamma}{\sqrt{\lambda_{2, \gamma}}}a_j 
\Bigg] \nonumber \\
& \le & J_L \Big[  \sqrt{\frac{\lambda_{2,\gamma}}{\lambda_{2,\gamma}} } \phi\left( \frac{a_j \dot{h}_\gamma(t^*)}{\sqrt{\lambda_{2,
\gamma}}}\right)\vert S_{j,\gamma}\setminus S_j \vert +  \phi\left( \frac{a_j \dot{h}_\gamma(g^*)}{\sqrt{\lambda_{2,
\gamma}}}\right)\frac{\Delta \dot{h}_\gamma}{\sqrt{\lambda_{2, \gamma}}}a_j 
\Big]\nonumber \\
& \le &   L \sqrt{\frac{\lambda_{2,\gamma}}{\lambda_{2,\gamma}} } \phi\left( \frac{a_j \dot{h}_\gamma(t^*)}{\sqrt{\lambda_{2,
\gamma}}}\right) +  \frac{J_L}{L}L a_j\phi\left( \frac{a_j \dot{h}_\gamma(g^*)}{\sqrt{\lambda_{2,
\gamma}}}\right)\frac{\Delta \dot{h}_\gamma}{\sqrt{\lambda_{2, \gamma}}}.  \nonumber
\end{eqnarray}
This bound goes to 0 by the lemma's conditions.

\item
Immediate  from the previous part.
\end{enumerate}
\end{proof}

\begin{proof}[{\bf Proof of Theorem \ref{thm:strongFWER}}]
\hfill\par\noindent
\begin{enumerate}
\item
By Proposition \ref{thm:p-values} and Lemma \ref{lemma:m_u},
\begin{eqnarray}
\label{eq:sFWER}
& & \FWER(u^*_{\Bon})  =  \P\left(V(u^*_{\Bon})\ge 1\right) = 1-
\P\left(V(u^*_{\Bon})= 0 \right) \nonumber \\ 
& = & 1 - \P \left( \Big[ \#\{t \in \tilde{T}\cap \mathbb{S}_{0, \gamma}:
\max x_\gamma(t)>u^*_{\Bon} \} + \#\{t \in \tilde{T}\cap(
\mathbb{S}_0\setminus  \mathbb{S}_{0, \gamma}):
\max x_\gamma(t)>u^*_{\Bon} \} \Big] =0 \right) \nonumber \\
& = & 1-\P\left(V_\gamma(u^*_{\Bon}) = 0 ~\text{and}~ \# \{t\in \tilde{T}\cap
( \mathbb{S}_0\setminus  \mathbb{S}_{0, \gamma}):\max x_\gamma(t)>u^*_{\Bon} \} = 0 \right)  \nonumber \\
& \le &  \P\left(V_\gamma(u^*_{\Bon})\ge 1\right) +  \P\left(\# \{t\in \tilde{T}\cap
( \mathbb{S}_0\setminus  \mathbb{S}_{0, \gamma})\}\ge 1 \right).  
\end{eqnarray}
The last inequality holds since $1-\P(A\cap B) \le \P(A^C)+ \P(B^C)$
and  
$$
 \P\left(\#\{t \in \tilde{T}\cap( \mathbb{S}_0\setminus  \mathbb{S}_{0\gamma}):
\max x_\gamma(t)>u^*_{\Bon} \}  \ge 1 \right) \le  \P\left(\# \{t\in \tilde{T}\cap
( \mathbb{S}_0\setminus  \mathbb{S}_{0, \gamma})\}\ge 1 \right).  
$$
The second probability in \eqref{eq:sFWER} goes to 0  by Lemma \ref{lemma:transitionE}, part (2).

Although the process $x_\gamma(t)$ is not stationary under the true
model, it has the same properties as  stationary process on the set
$\mathbb{S}_{0, \gamma}$, since $x_\gamma(t) = z_\gamma(t)$ for all $t \in \mathbb{S}_{0, \gamma}$.
Thus, following the same arguments as in the first part of Theorem
\ref{thm:p-values} we get that the first probability in
\eqref{eq:sFWER} is bounded by $ \alpha$. 
The second probability goes to 0 by Lemma \ref{lemma:transitionE}.

\item
Let  $\tilde{m}_{0, \gamma}$ be the number of local maxima belonging
to the set $ \mathbb{S}_{0, \gamma}$ and let $\tilde{v}_\gamma = F_\gamma^{-1}\left(\alpha/\tilde{m}_{0, \gamma}\right)$.
It is clear that $\tilde{v}_\gamma \le \tilde{u}_{\Bon}$ since
$\tilde{m}_{0, \gamma} \le \tilde{m}$, therefore 
\begin{eqnarray}
\label{eq:sFWER2}
\FWER(\tilde{u}_{\Bon}) & = & \P\left(V(\tilde{u}_{\Bon})\ge1 \right) \le
\P\left(V(\tilde{v}_\gamma)\ge1 \right)  \nonumber \\
& \le & \P\left(V_\gamma(\tilde{v}_\gamma)\ge 1\right) +
  \P\left(\# \{t\in \tilde{T}\cap
( \mathbb{S}_0\setminus  \mathbb{S}_{0, \gamma}) \}\ge 1 \right) 
\end{eqnarray}
by an argument similar to that in \eqref{eq:sFWER}.
Following the same arguments as in the proof of Theorem
\ref{thm:Bonferroni} part (2), the first probability in
\eqref{eq:sFWER2} is bounded by $\alpha$ as $L \to \infty$.
The second probability goes to 0 by Lemma \ref{lemma:transitionE}.
\end{enumerate}
\end{proof}

\subsection{Proof of Theorem \ref{thm:FDRcontrol} (Control of FDR)} 
\label{App:FDR}
\begin{lemma}
\label{lemma:unique-max}
Recall from Section \ref{sec:model} and \ref{sec:proc} that $\mu_\gamma(t) = a_jh_\gamma(t-\tau_{j, \gamma})$ with finite support
$S_{j, \gamma}$ and $h_\gamma(t)$ has a unique local maximum at $
t=\tau_{j, \gamma}$ that is an interior point of $S_j \subset S_{j,
\gamma}$ and has no other critical points.
Let $0 < \eps < 1$ and $I_{j, \eps} = [\tau_{j, \gamma} - \eps , \tau_{j, \gamma} + \eps
] \subset S_{j} \subset S_{j, \gamma}$.
Then, as $a_j \to \infty$ and $\eps \to 0$ such that $a_j \dot{h}_\gamma(-\eps) \to \infty$ and $a_j \dot{h}_\gamma(\eps) \to -\infty$,
\begin{enumerate}
\item $\P\left(\# \{t \in \tilde{T}\cap I_{j, \eps}\} \ge
1\right) \to 1$ 
\item $\P\left(\# \{t \in \tilde{T}\cap I_{j, \eps}:x_\gamma(t)>u \}\ge
1\right) \to 1$  for every fixed threshold $u$.
\end{enumerate}
\end{lemma}


\begin{proof}
\hfill\par\noindent
\begin{enumerate}
\item  The probability that $x_\gamma(t)$ has some local maxima in $I_\eps$ is
\begin{eqnarray}
\label{eq:I_bound}
1 & \ge &  \P\left(\# \{t \in \tilde{T}\cap I_{j, \eps}\} \ge
1\right) \ge 
\P\left\{\dot{x}_\gamma(\tau_{j, \gamma} - \eps) > 0 ~\text{and}~
\dot{x}_\gamma(\tau_{j, \gamma} + \eps) < 0 \right\} \nonumber \\
& = & \P\left\{\dot{z}_\gamma(\tau_{j, \gamma} - \eps) > -\dot{\mu}_\gamma(\tau_{j, \gamma} - \eps) ~\text{and}~
\dot{z}_\gamma(\tau_{j, \gamma} + \eps) < -\dot{\mu}_\gamma(\tau_{j,
\gamma} + \eps) \right\} \nonumber \\
& \ge & 1-\left[ \P\left(\dot{z}_\gamma(\tau_{j, \gamma} - \eps) \le
-\dot{\mu}_\gamma(\tau_{j, \gamma} - \eps)\right) +
\P\left(\dot{z}_\gamma(\tau_{j, \gamma} + \eps) \ge  -\dot{\mu}_\gamma(\tau_{j,
\gamma} + \eps)\right) \right] \nonumber \\
& \ge &  1-\left[ \P\left(\inf_{[\tau_{j, \gamma} - \eps,
\tau_{j,\gamma}]} \dot{z}_\gamma(t) \le -\dot{\mu}_\gamma(\tau_{j,
\gamma} - \eps)\right)  + \P\left(\sup_{[\tau_{j, \gamma}
,\tau_{j,\gamma}+ \eps]}\dot{z}_\gamma(t) \ge  -\dot{\mu}_\gamma(
\tau_{j, \gamma} + \eps )\right) \right] \nonumber \\
& = &  1-\left[ \P\left(\sup_{[\tau_{j, \gamma} - \eps,
\tau_{j,\gamma}]} -\dot{z}_\gamma(t) \ge a_j\dot{h}_\gamma(- \eps)\right)  + \P\left(\sup_{[\tau_{j, \gamma}
,\tau_{j,\gamma}+ \eps]}\dot{z}_\gamma(t) \ge  -a_j\dot{h}_\gamma(\eps )\right) \right].
\end{eqnarray}

The probability that the supremum of any differentiable random process, $f(t)$, is above $u$ is bounded  by \citep{Adler:2007}.
$$
\P\left(\sup_{t\in [0,T]}f(t) \ge u\right) \le   \P\left(f(0) \ge u\right) + \E[N_u] \triangleq  G(u,|T|,\sigma^2),
$$
where $N_u$ is the number of up-crossings by $f$ of the level $u$ in the interval $[0, T]$ and $\sigma^2$ is the variance of the process. For the stationary Gaussian process $\dot{z}_\gamma(t)$ with mean 0 and variance $\lambda_{2, \gamma}$, applying Kac-Rice formula gives:
$$
\E[N_u] = |T|\,p(\dot{z}_\gamma = u)\int_0^\infty xp(\ddot{z}_\gamma = x)\,dx
= \frac{|T|}{2\pi} \exp{\left(-\frac{u^2}{2\lambda_{2, \gamma}}\right)} \sqrt{
\frac{\lambda_{4, \gamma} }{ \lambda_{2, \gamma}} }.
$$
In particular,
\begin{eqnarray}
\label{eq:sup_bound}
 \P\left(\sup_{[\tau_{j, \gamma} - \eps, \tau_{j, \gamma}] } -\dot{z}_\gamma(t) \ge u \right) &
\le & 
1-\Phi\left(\frac{u}{\sqrt{\lambda_{2, \gamma}}}\right) + \frac{\eps}{2\pi}
\exp{\left(-\frac{u^2}{2\lambda_{2, \gamma}}\right)} \sqrt{
\frac{\lambda_{4, \gamma} }{ \lambda_{2, \gamma}} } \nonumber \\ 
& \le & \frac{\sqrt{\lambda_{2, \gamma}}}{u}
\phi\left(\frac{u}{\sqrt{\lambda_{2, \gamma}}}  \right) +  \eps\sqrt{
\frac{\lambda_{4, \gamma} }{2\pi \lambda_{2,
\gamma}}}\phi\left(\frac{u}{\sqrt{\lambda_{2, \gamma}}} \right)
\nonumber \\
& = & G(u,\eps, \lambda_{2,\gamma} ).
\end{eqnarray}

The bracketed expression in \eqref{eq:I_bound} can be bounded by
\begin{eqnarray}
\label{eq:upper-bound}
0 & \le &  \P\left(\sup_{[\tau_{j, \gamma} - \eps,
\tau_{j,\gamma}]} -\dot{z}_\gamma(t) \ge a_j\dot{h}_\gamma(- \eps)\right)  + \P\left(\sup_{[\tau_{j, \gamma}
,\tau_{j,\gamma}+ \eps]}\dot{z}_\gamma(t) \ge  -a_j\dot{h}_\gamma(
\eps )\right)  \nonumber \\
& \le & G\left(a_j\dot{h}_\gamma(- \eps), \eps, \lambda_{2, \gamma} \right) + G\left(-a_j\dot{h}_\gamma(\eps), \eps, \lambda_{2, \gamma} \right).
\end{eqnarray}
Both terms in \eqref{eq:upper-bound} go to 0. Thus back to \eqref{eq:I_bound}, $\P(\# \{t \in \tilde{T}\cap I_{j, \eps}\} \ge 1) \to 1$.

\item
The probability that at least one local maxima in $I_{j, \eps}$
exceeds the fixed threshold $u$ satisfies
$$
\begin{aligned}
& \P\left(\# \{t \in \tilde{T}\cap I_{j,
\eps}:x_\gamma(t)>u \}\ge 
1\right) \ge \P\left( \inf_{I_{j, \eps}}
x_\gamma(t) >u \right) =  \P\left( \inf_{I_{j, \eps}}
\left[\mu_\gamma(t) + z_\gamma(t)\right] >u \right)  \\ &\ge  \P\left(
\inf_{I_{j, \eps}}\mu_\gamma(t) + \inf_{I_{j, \eps}}z_\gamma(t) > u
\right) = \P\left( \inf_{I_{j, \eps}}z_\gamma(t) > u -
a_j\inf_{I_{j, \eps}}h_\gamma(t) )\right) \\ &=  \P\left( \sup_{I_{j,
\eps}}-z_\gamma(t) < a_j\inf_{I_{j, \eps}}h_\gamma(t) - u \right) = 
1- \P\left( \sup_{I_{j,
\eps}}-z_\gamma(t) \ge  a_j\inf_{I_{j, \eps}}h_\gamma(t) - u \right). 
\end{aligned}
$$
The infimum of $h_\gamma(t)$ in the range $I_{j, \eps}$ occurs at 
one of the edges. Without loss of generality we assume that
$\inf_{I_{j, \eps}}h_\gamma(t) = h_\gamma(\eps)$. Using the upper
bound given in equation \eqref{eq:sup_bound}, applied to the interval $I_{j, \eps}$ and $z_\gamma(t)$ rather than $\dot{z}_\gamma(t)$, yields
\begin{equation}
\label{eq:exceedU}
\P\left(\# \{t \in \tilde{T}\cap I_{j,\eps}:x_\gamma(t)>u \}\ge 
1\right) \ge 1-C_{\eps}(u),
\end{equation}
where
$$
C_\eps(u) = \frac{\sigma_\gamma}{a_j h_\gamma(\eps) - u }\phi\left(\frac{a_j h_\gamma(\eps) - u}{\sigma_\gamma}\right) + 2\eps\sqrt{\frac{\lambda_{2,\gamma}}{2\pi\sigma^2_{\gamma}}}\phi\left(\frac{a_j h_\gamma(\eps) - u}{\sigma_\gamma}\right).
$$
For any constant $u$ the convergence rate of $a_jh_\gamma(\eps)-u$ is the same as that of $a_jh_\gamma(\eps)$. Therefore, $C_\eps(u) \to 0$ as $a_j \to \infty$ for any fixed threshold $u$.
\end{enumerate}
\end{proof}

\begin{lemma}
\label{lemma:results2}
Assume the model of Section \ref{sec:model} and the procedure of Section \ref{sec:proc}. Let $\tilde{m}_{1, \gamma} = \#\{ \tilde{T}\cap \mathbb{S}_{1, \gamma} \}$ be the number of local maxima in the set $\mathbb{S}_{1, \gamma}$ and recall that $W_\gamma(u) = \#\{t\in \tilde{T}\cap \mathbb{S}_{1, \gamma}:x_\gamma( t)> u \}$ is the number of local maxima in $\mathbb{S}_{1, \gamma}$ above threshold $u$. Under the assumptions of Theorem \ref{thm:FDRcontrol},
\begin{enumerate}
\item
The probability to get exactly $J_L$ local maxima in the 
set $\mathbb{S}_{1, \gamma}$, $\P\left( \tilde{m}_{1, \gamma} = J_L\right)  = \P\left(\#\{ \tilde{T}\cap \mathbb{S}_{1, \gamma} \} = J_L\right)$ tends to 1.

\item The probability to get exactly $J_L$ local maxima in the 
set $\mathbb{S}_{1, \gamma}$ that exceed the threshold $u$,
$$
\P\left( W_\gamma(u) = J_L\right) = \P\left(\# \{t\in \tilde{T}\cap \mathbb{S}_{1, \gamma}:x_\gamma(t)> u \}=J_L \right) \to 1
$$
tends to 1, for any fixed threshold $u$.
\item $\tilde{m}_{1, \gamma}/L \to A_1$ in probability.

\item $W_\gamma(u)/\tilde{m}_{1, \gamma} \to 1$ in probability.

\end{enumerate}
\end{lemma}

\begin{proof}
\hfill\par\noindent
\begin{enumerate}
\item
The probability to get exactly $J_L$ local maxima in the set
$\mathbb{S}_{1, \gamma}$ is
\begin{equation}
\begin{aligned}
\label{eq:Jmaxima}
1 &\ge \P \left( \# \{t \in \tilde{T}\cap \mathbb{S}_{1, \gamma}   \}=J_L
\right)
= \P \left( \# \{t \in \tilde{T}\cap (\cup_{j=1}^{J_L} S_{j, \gamma} )  \}=J_L
\right) \\
&\ge  \P \left(\cap_{j=1}^{J_L} \{ \# \{t \in \tilde{T}\cap S_{j, \gamma}   \}=1\}
\right) \\
&\ge \P \left(\cap_{j=1}^{J_L}  \# \{t \in \tilde{T}\cap I_{j, \eps}   \}=1
\right) , \qquad I_{j, \eps} \subset S_{j, \gamma} \\
&\ge 1- \sum_{j=1}^{J_L} \left[ 1- \P \left(\# \{t \in
\tilde{T}\cap I_{j, \eps}   \}=1
\right) \right].\\
\end{aligned}
\end{equation}
Combining  equations \eqref{eq:I_bound} and \eqref{eq:upper-bound},  the above expression is bounded below by
\begin{eqnarray}
& & 1-\sum_{j=1}^{J_L} \left(G\left(a_j\dot{h}_\gamma(-\eps), \eps,\lambda_{2, \gamma} \right) + G\left(-a_j\dot{h}_\gamma(\eps), \eps,\lambda_{2, \gamma}  \right)\right) \nonumber \\
& \ge & 1- J_L \left(G\left(a^*\dot{h}_\gamma(-\eps) , \eps,\lambda_{2, \gamma} \right) + G\left(-a^*\dot{h}_\gamma(\eps), \eps,\lambda_{2, \gamma}  \right)\right) \nonumber \\
& \ge & 1 -
\frac{J_L}{L} L\left[ \frac{\sqrt{\lambda_{2, \gamma}}}{a^*\dot{h}_\gamma(-\eps)}\phi\left( \frac{a^*\dot{h}_\gamma(-\eps)}{\sqrt{\lambda_{2, \gamma}}}\right) +  \frac{\sqrt{\lambda_{2, \gamma}}}{-a^*\dot{h}_\gamma(\eps)}\phi\left( \frac{a^*\dot{h}_\gamma(\eps)}{\sqrt{\lambda_{2, \gamma}}}\right) \right] \nonumber \\
& &  - ~L\sqrt{\frac{\lambda_{4, \gamma}}{2\pi\lambda_{2, \gamma}}}\left[\phi\left( \frac{a^*\dot{h}_\gamma(\eps)}{\sqrt{\lambda_{2, \gamma}}}\right)+ \phi\left( \frac{a^*\dot{h}_\gamma(-\eps)}{\sqrt{\lambda_{2, \gamma}}}\right) \right], \label{eq:Jmax-bound}
\end{eqnarray}
where $a^* = \arg\max \{G(a_j\dot{h}_\gamma(-\eps), \eps,\lambda_{2, \gamma} ) + G(-a_j\dot{h}_\gamma(\eps), \eps,\lambda_{2, \gamma} )\}$.
 
The fastest rate of convergence toward zero, for which the requirement is still valid, is achieved when $\dot{h}_\gamma(\pm\eps)$ converges to zero as $a^{-(1-\delta)}, ~ 0 <\delta <1$. The lemma's conditions guarantee that expressions in \eqref{eq:Jmax-bound} go to 0 in the least favourable case.
\item
Following the same computations as in equation \eqref{eq:Jmaxima} and using the bound \eqref{eq:exceedU}, the probability to have $J_L$ truly  rejected hypotheses is bounded by
\begin{eqnarray}
1 & \ge &  \P\left(\# \{t\in \tilde{T}\cap
\mathbb{S}_{1, \gamma}:x_\gamma(t)> u \}=J_L \right) \nonumber  \\
& \ge & 1- \left(\sqrt{\sigma_\gamma}\frac{J_L}{a_jh_\gamma(\eps)
- u}\phi\left(\frac{a_j h_\gamma(\eps) - u}{\sigma_\gamma}\right) +  L\sqrt{\frac{\lambda_{2,
\gamma}}{2\pi\sigma^2_\gamma}}\phi\left(\frac{a_j h_\gamma(\eps) - u}{\sigma_\gamma}\right)\right). \nonumber
\end{eqnarray}
Following the same arguments as those in the last paragraph in the first part of this lemma, the expression in parentheses goes to 0 for any fixed threshold $u$.

\item
Since
$$
\begin{aligned}
\frac{\tilde{m}_{1, \gamma}}{L} =\frac{\tilde{m}_{1, \gamma}}{J_L}
\frac{J_L}{L},
\end{aligned}
$$
we need to show that $\tilde{m}_1/J_L \to 1$ in probability. For any fixed $\eps>0$
$$
\begin{aligned}
0 \le \P\left(\left|\frac{\tilde{m}_{1, \gamma}}{J_L}-1 \right| \ge \eps \right) 
= \P\left( \left|\tilde{m}_{1, \gamma} - J_L\right| \ge J_L\eps\right) \le
\P\left( \tilde{m}_{1, \gamma} \neq J_L\right) = 1-\P\left(\tilde{m}_{1, \gamma} = J_L \right)
\end{aligned}
$$
since $\tilde{m}_{1, \gamma}$ and $J_L$ are integers. The probability to get exactly $J_L$ local maxima goes to 1 by the part (1) of this lemma.

\item
By part (2) of this lemma $\P(W_\gamma(u) = J_L)\to 1$ in probability, therefore, using the same arguments as in part (3) of this lemma, we get $W_\gamma(u)/J_L \to 1$. Now,
$$
\frac{W_\gamma(u)}{\tilde{m}_{1, \gamma}} =
\frac{W_\gamma(u)}{J_L}\frac{J_L}{\tilde{m}_{1, \gamma}}.
$$
Both fractions go to 1 by the previous parts of this lemma. 
\end{enumerate}
\end{proof}

\begin{lemma}
\label{lemma:mFDR}
For any fixed threshold $u$, and any positive integer $J$
$$
\FDR(u) \le \P(W(u) \le J-1) + \frac{\E[V(u)]}{\E[V(u)]+J },
$$
where  $W(u) = \# \{t\in \tilde{T} \cap \mathbb{S}_1: x_\gamma(t)> u
\}$. 
\end{lemma}

\begin{proof}
For simplicity, in this proof we omit the argument $u$, since it is fixed.
$$
\begin{aligned}
\FDR &= \E \left(\frac{V} {V+W} \right)  = \sum_{v=0}^\infty 
\sum_{w=0}^\infty \left(\frac{v}{v+w} \right)\P\left(V=v, W=w \right)
\\
&=  \sum_{v=0}^\infty 
\sum_{w=0}^{J-1} \left(\frac{v}{v+w} \right)\P\left(V=v, W=w \right) +  \sum_{v=0}^\infty 
\sum_{w=J}^\infty \left(\frac{v}{v+w} \right)\P\left(V=v, W=w \right)
\\
&\le  \sum_{w=0}^{J-1}\sum_{v=0}^\infty \P\left(V=v, W=w \right) +  \sum_{v=0}^\infty 
\sum_{w=J}^\infty \left(\frac{v}{v+J} \right)\P\left(V=v, W=w \right)
\\
&=  \sum_{w=0}^{J-1}\P\left( W=w \right) +  \sum_{v=0}^\infty \left(\frac{v}{v+J} \right) 
\P\left(V=v, W\ge J \right)\\
&\le  \P\left( W \le J-1 \right) +  \sum_{v=0}^\infty \left(\frac{v}{v+J} \right) \P\left(V=v \right)\\
&=  \P\left( W \le J-1 \right) +  \E\left(\frac{V}{V+J} \right)
\le\P\left( W \le J-1 \right) +  \frac {\E(V)}{\E(V)+J}.
\end{aligned}
$$
The last inequality holds by Jensen's inequality, since $V/(V+J)$ is a concave
function of $V$ for $V \ge 0$ and $J \ge 1$.
\end{proof}

\begin{proof}[{\bf Proof of Theorem \ref{thm:FDRcontrol}}]
Let $R(\tilde{u}_{\BH}) = \tilde{m}-k+1$ be the number of rejected
hypotheses using the threshold $\tilde{u}_{\BH}$. 
Let 
$$
\tilde{G}(u) = \frac{\# \{x_\gamma(t)>u, \dot{x}_\gamma(t)=0,
\ddot{x_\gamma}(t)<0 \}} {\#\{ \dot{x}_\gamma(t)=0,
\ddot{x}_\gamma(t)<0 \}}
$$
be the empirical marginal right cumulative distribution function of
$x_\gamma(t)|t \in \tilde{T}$.
Then
$$
\tilde{G}(\tilde{u}_{\BH}) = \frac{\#
\{x_\gamma(t)>\tilde{u}_{\BH}, \dot{x}_\gamma(t)=0, 
\ddot{x_\gamma}(t)<0 \}} {\#\{ \dot{x}_\gamma(t)=0,
\ddot{x}_\gamma(t)<0 \}} = \frac {R(\tilde{u}_{\BH})}
{\tilde{m}}.
$$
Multiplying both sides of this equation by $\alpha$ we get
$$
\alpha \tilde{G}(\tilde{u}_{\BH})  = \alpha  \frac {R(\tilde{u}_{\BH})}
{\tilde{m}} =  F_\gamma(\tilde{u}_{\BH}).
$$
This result implies that the BH threshold $\tilde{u}_{\BH}$ is the largest $u$ that solves the equation
\begin{equation}
\label{eq:FDRthreshold}
\alpha \tilde{G}(u)  = F_\gamma(u).
\end{equation}

The strategy is to solve equation \eqref{eq:FDRthreshold} in the limit when $L, a_j \to \infty$. We first find the limit of $\tilde{G}(u)$. Letting  $W_\gamma(u) = \#\{t \in \tilde{T} \cap\mathbb{S}_{1,
\gamma}: x_\gamma(t)>u \}$ and $V_\gamma(u) = \#\{t \in \tilde{T} \cap\mathbb{S}_{0,
\gamma}: x_\gamma(t)>u \}$ ,
\begin{equation}
\label{eq:ecdf}
\begin{aligned}
\tilde{G}(u) &= \frac{\# \{x_\gamma(t)>u, \dot{x}_\gamma(t)=0,
\ddot{x}_\gamma(t)<0 \}} {\#\{ \dot{x}_\gamma(t)=0,
\ddot{x}_\gamma(t)<0 \}} = \frac {V_\gamma(u)+ W_\gamma(u)}
{\tilde{m}_0+\tilde{m}_1} \\
&=\frac{V_\gamma(u)} {\tilde{m}_{0, \gamma}}
\frac{\tilde{m}_{0, \gamma}}{\tilde{m}_{0, \gamma} +\tilde{m}_{1,
\gamma}} +\frac{W_\gamma(u)} {\tilde{m}_{1, \gamma}}
\frac{\tilde{m}_{1, \gamma}}{\tilde{m}_{0, \gamma} +\tilde{m}_{1, \gamma}} \\
\end{aligned}
\end{equation}
Recall that $x_\gamma(t)$ is ergodic, thus by the weak law of large
numbers and Lemma \ref{lemma:results2} part (3), 
\begin{equation}
\label{eq:G1}
\frac{\tilde{m}_{0, \gamma}}{\tilde{m}_{0, \gamma} +\tilde{m}_{1,
\gamma}}  = \frac{\tilde{m}_{0, \gamma}/L}{\tilde{m}_{0, \gamma}/L
+\tilde{m}_{1, \gamma}/L} \to
\frac{\E[\tilde{m}_{0, \gamma}; [0, 1]]}{\E[\tilde{m}_{0,\gamma}; [0, 1]] + A_1}, ~\text{as}~ L\to \infty.
\end{equation}
Following the  arguments  in Lemma \ref{lemma:m_u}, for $L\to \infty$,
\begin{equation}
\label{eq:G2}
\frac{V_\gamma(u)} {\tilde{m}_{0, \gamma}} \to \frac{\E[V_\gamma(u)]}
{\E[\tilde{m}_{0, \gamma}]} = F_\gamma(u).
\end{equation}
Finally, by Lemma \ref{lemma:results2}, part (4)
\begin{equation}
\label{eq:G3}
\frac{W_\gamma(u)} {\tilde{m}_{1, \gamma}} \to 1
\end{equation}
as $L \to \infty$ and $a_j \to \infty$ such that $L\phi(Ka_j^\delta)
\to 0$.

Combining equations \eqref{eq:G1}, \eqref{eq:G2} and \eqref{eq:G3}
with Lemma \ref{lemma:results2} part (4) in \eqref{eq:ecdf}, we obtain
$$
 \tilde{G}(u) \to  F_\gamma(u)\frac{\E[\tilde{m}_{0, \gamma}; [0,
1]]}{\E[\tilde{m}_{0, \gamma}; [0, 1]] +
A_1} + \frac{A_1}{\E[\tilde{m}_{0, \gamma}; [0, 1]] + A_1}.
$$
Now replacing $\tilde{G}(u)$ by its limit in \eqref{eq:FDRthreshold}, we obtain
$$
\alpha \left( F_\gamma(u)\frac{\E[\tilde{m}_{0, \gamma}; [0,
1]]}{\E[\tilde{m}_{0, \gamma}; [0, 1]] +
A_1} + \frac{A_1}{\E[\tilde{m}_{0, \gamma}; [0, 1]] + A_1}\right) = F_\gamma(u),
$$
leading to the solution
\begin{equation}
\label{eq:eq}
F_\gamma(u) = \frac {\alpha A_1}{A_1 + \E[\tilde{m}_{0, \gamma}; [0, 1]](1-\alpha)}.
\end{equation}
Note that $F_\gamma(u)$ is convex for $u>0$, therefore exists a unique solution to  equation \eqref{eq:eq}. Let $u^*_{\BH}$ be the solution of equation \eqref{eq:eq}. 

The FDR at the threshold $u^*_{\BH}$ is bounded by Lemma \ref{lemma:mFDR} by
$$
\begin{aligned}
\FDR(u^*_{\BH}) &\le \P(W(u^*_{\BH})\le J_L-1) +  \frac{ \E\left[ V(u^*_{\BH})\right]}{ \E\left[
V(u^*_{\BH})\right] + J_L } \\
 &= \P(W(u^*_{\BH})\le J_L-1) +  \frac{
\E\left[ V_\gamma(u^*_{\BH})\right] +  \E\left[ \#\{t \in
\tilde{T}\cap\left(\mathbb{S}_0\setminus \mathbb{S}_{0, \gamma} 
\right): x_\gamma(t)>u^*_{\BH} \} \right]  }{ \E\left[
V_\gamma(u^*_{\BH})\right] +  \E\left[ \#\{t \in
\tilde{T}\cap\left(\mathbb{S}_0\setminus \mathbb{S}_{0, \gamma}
\right): x_\gamma(t)>u^*_{\BH} \} \right] + J_L }, 
\end{aligned}
$$
where we have split $V_\gamma(u^*_{\BH})$ into the region $\mathbb{S}_{0, \gamma}$ and the transition region $\mathbb{S}_0 \setminus \mathbb{S}_{0, \gamma}$.

When both $L$ and $a_j$ go to infinity such that $L\phi(Ka_j) \to 0$,
Lemma \ref{lemma:results2} gives
\begin{equation}
\label{eq:lim2}
0 \le \E\left[ \#\{t \in
\tilde{T}\cap\left(\mathbb{S}_0\setminus \mathbb{S}_{0, \gamma} 
\right): x_\gamma(t)>u^*_{\BH} \} \right] \le \E\left[ \#\{t \in
\tilde{T}\cap\left(\mathbb{S}_0\setminus \mathbb{S}_{0, \gamma} 
\right)\} \right] \to 0.
\end{equation}
By \eqref{eq:G2}, the remaining terms of the fraction can be written as
$$
\begin{aligned}
\frac{ \E\left[ V_\gamma(u^*_{\BH})\right]}{ \E\left[     
V_\gamma(u^*_{\BH})\right]+J_L} = 
\frac{F_\gamma(u^*_{\BH})\E[\tilde{m}_{0, \gamma}; [0, 1]]L}
{F_\gamma(u^*_{\BH})\E[\tilde{m}_{0, \gamma}; [0, 1]]L + J_L}   =
\frac{F_\gamma(u^*_{\BH})\E[\tilde{m}_{0, \gamma}; [0, 1]]}
{F_\gamma(u^*_{\BH})\E[\tilde{m}_{0, \gamma}; [0, 1]] + J_L/L}. 
\end{aligned}
$$
Since $u^*_{\BH}$ solves \eqref{eq:eq}, for $L \to \infty$ such that $ J_L/L \to A_1$ the above expression tends to
\begin{eqnarray}
\label{eq:lim1}
\frac{\alpha\E[\tilde{m}_{0, \gamma}; [0, 1]]}{\alpha\E[\tilde{m}_{0, \gamma}; [0, 1]] +
A_1 + (1-\alpha)\E[\tilde{m}_{0, \gamma}]; [0, 1]}  =
\alpha\frac{\E[\tilde{m}_{0, \gamma}; [0, 1]]}{\E[\tilde{m}_{0, \gamma}; [0, 1]] + A_1} \le \alpha.
\end{eqnarray}
Combining equations \eqref{eq:lim2}, \eqref{eq:lim1} and Lemma
\ref{lemma:results2} part (3), we obtain
$$
\lim \sup \FDR(u^*_{\BH}) \le \alpha
$$
as $L$ and $a_j$ go to infinity such that $L\phi(Ka_j^\delta) \to 0$.

Recall that the BH threshold $\tilde{u}_{\BH}$ solves the equation
\eqref{eq:FDRthreshold}, and $u^*_{\BH}$ solves the equation
\eqref{eq:eq}, where the empirical marginal distribution, $\tilde{G}(u)$, is replaced by its limit. Since $F_\gamma(t)$ is continuous 
\begin{equation}
\label{eq:FDRthreshold-limit}
F_\gamma(\tilde{u}_{\BH})\to F_\gamma(u^*_{\BH}),
\end{equation}
leading to
$$
\lim \sup \FDR(\tilde{u}_{\BH}) \le \alpha.
$$

\end{proof}

\subsection{Power}
\begin{proof}[{\bf Proof of Lemma \ref{lemma:fwer-fdr-threshold}}]
\hfill\par\noindent
\begin{enumerate}
\item Let us show first that $F_\gamma(a_jh_\gamma(0))/F_\gamma(u^*_{\Bon}) \to 0$. To show this, we bound $F_\gamma(x)$ below by
\begin{equation}
\label{eq:F-lbound}
F_\gamma(x) \ge \sqrt{\frac{2\pi\lambda_{2,\gamma}^2}{\lambda_{4,\gamma} \sigma^2_\gamma}}\phi\left(\frac{x}{\sigma_\gamma}\right)\frac{1}{2} = C_1 \phi\left(\frac{x}{\sigma_\gamma}\right), \qquad
C_1 = \sqrt{\frac{\pi\lambda_{2,\gamma}^2}{2\lambda_{4,\gamma} \sigma^2_\gamma}}.
\end{equation}
An upper bound is given by
$$
F_\gamma(x) \le 1-\Phi\left(\frac{x}{\sigma_\gamma}\right) + \sqrt{\frac{2\pi\lambda_{2,\gamma}^2}{\lambda_{4,\gamma} \sigma^2_\gamma}}\phi\left(\frac{x}{\sigma_\gamma}\right)
\le \phi\left(\frac{x}{\sigma_\gamma}\right)\left[ \sqrt{\frac{2\pi\lambda_{2,\gamma}^2}{\lambda_{4,\gamma} \sigma^2_\gamma}} + \frac{\sigma_\gamma}{x}\right],
$$
where the first inequality uses the fact that $\lambda_{4,\gamma}/\Delta \ge 1/\sigma^2_\gamma$, and the second follows from \eqref{eq:Gaus-ineq}. Assuming that $x$ is large enough, say $x>\sigma_\gamma$, the upper bound takes the form
\begin{equation}
\label{eq:F-ubound}
F_\gamma(x) \le C_2\phi\left(\frac{x}{\sigma_\gamma}\right), \qquad
C_2 = \sqrt{\frac{2\pi\lambda_{2,\gamma}^2}{\lambda_{4,\gamma} \sigma^2_\gamma}} + 1.
\end{equation}
Thus
\begin{equation}
\label{eq:Ffrac}
\frac{F_\gamma(a_jh_\gamma(0))}{F_\gamma(u_{\Bon})} \le 
 C_2\frac{1}{\alpha}\frac{\E[\tilde{m}]}{L} L\phi\left(\frac{a_jh_\gamma(0)}{\sigma_\gamma}\right),
\end{equation}
which goes to 0 by the conditions of the lemma and by the fact that $\E[\tilde{m}]/L = \E[\tilde{m}; [0, 1]]$ does not depend on $L$.

Define now $v = F_\gamma(u^*_{\Bon})$ and $ w=F_\gamma(a_jh_\gamma(0))$.
Inverting the bounds \eqref{eq:F-lbound} and \eqref{eq:F-ubound}, we get bounds
\begin{eqnarray*}
\label{eq:u-bon}
C_1 \phi\left(\frac{u^*_{\Bon}}{\sigma_\gamma}\right) \le &  v  & \le C_2 \phi\left(\frac{u^*_{\Bon}}{\sigma_\gamma}\right) \nonumber \\
\ln{\frac{C_1}{\sqrt{2\pi}}} - \frac{(u^*_{\Bon})^2}{2\sigma^2_\gamma} \le &  \ln{v}  & \le \ln{\frac{C_2}{\sqrt{2\pi}}} - \frac{(u^*_{\Bon})^2}{2\sigma^2_\gamma} \nonumber \\
2\sigma^2_\gamma \left(\ln{\frac{C_1}{\sqrt{2\pi}}} - \ln v\right) \le & (u^*_{\Bon})^2 & \le 2\sigma^2_\gamma \left(\ln{\frac{C_2}{\sqrt{2\pi}}} - \ln v\right) \nonumber
\end{eqnarray*}
for $u^*_{\Bon}$ and similarly for $a_jh_\gamma(0)$. Thus,
\begin{eqnarray}
\label{eq:u-over-a}
0 \le \frac{(u^*_{\Bon})^2}{(a_jh_\gamma(0))^2}  \le  \frac{\ln(C_2/\sqrt{2\pi}) - \ln(v)}{\ln(C_1/\sqrt{2\pi}) - \ln(w)}. \nonumber
\end{eqnarray}
Applying L'H\^{o}pital rule, the limit of the above fration when $v$ and $w$ go to zero is the same as the limit of $w/v$, which is 0 by \eqref{eq:Ffrac}.

\item In contrast to the Bonferroni threshold, the asymptotic FDR threshold depends only on the proportion of false null hypotheses. Because of the assumption
of asymptotically fixed proportion of true peaks, $J_L/L \to A_1$ with $0<A_1<1$, the limit of the FDR threshold \eqref{eq:eq} is fixed and finite. This leads to the result.
\end{enumerate}
\end{proof}

\begin{proof}[{\bf Proof of Theorem \ref{thm:power}}]
\hfill\par\noindent
Let $u$ be any fixed threshold and let $\eps$ and $I_{j, \eps}$ be defined as in Lemma \ref{lemma:unique-max}. Then, by the definition \eqref{eq:power} and the bound \eqref{eq:exceedU} in Lemma \ref{lemma:unique-max} ,
\begin{equation*}
\label{power-bound}
1 \ge \Power(u) \ge \P\left[ \#\{t \in I_{j, \eps} : x_\gamma(t)>u \} \ge 1\right] \ge 1-C_\eps(u)
\end{equation*}
Combining  the results of 
Lemmas \ref{lemma:fwer-fdr-threshold} and \ref{lemma:unique-max}, $C_\eps(u) \to 0$ and the power of Bonferroni and BH procedures when using the deterministic thresholds $u^*_{\Bon}$ and $u^*_{\BH}$ converges to 1 in probability.

It was proven earlier that the gap between the random and deterministic thresholds of the Bonferroni and BH procedures goes to 0 (see part (2) of Theorems \ref{thm:Bonferroni}, \ref{thm:strongFWER} and \eqref{eq:FDRthreshold-limit}). Therefore, the power when using random thresholds $\tilde{u}_{\Bon}$ and $\tilde{u}_{\BH}$ converges to 1 as well.
\end{proof}


\section*{Acknowledgments}

The authors thank Pablo Jadzinsky for providing the neural recordings data, and Igor Wigman, Felix Abramovich, and Yoav Benjamini for helpful discussions. This work was partially supported by NIH grant P01 CA134294-01, the Claudia Adams Barr Program in Cancer Research, and the William F. Milton Fund.

\bibliographystyle{plainnat}

\end{document}